\documentclass{amsart}
\usepackage{graphicx} 
\usepackage{url,multirow}
\usepackage{amsmath, amsthm, amsfonts, enumerate}
\usepackage{comment}
\usepackage{url,multirow}
\newtheorem{example}{Example}
\newtheorem{definition}{Definition}
\newtheorem{observation}{Observation}

\newtheorem{proposition}{Proposition}

\begin{document}
\author{David McCune}

\address{David McCune, Department of Mathematics and Data Science, William Jewell College, 500 College Hill, Liberty, MO, 64068}
\email{mccuned@william.jewell.edu} 

\author{Jennifer Wilson}
\address{Jennifer Wilson, The New School, Department of Natural Sciences and Mathematics, 66 West 12th Street New York, NY 10011}
\email{wilsonj@newschool.edu}

\title[IRV and the Reinforcement Paradox]{Instant Runoff Voting and the Reinforcement Paradox}

\begin{abstract}
We analyze the susceptibility of instant runoff voting (IRV) to a lesser-studied paradox known as a \emph{reinforcement paradox}, which occurs when candidate $X$ wins under IRV in two distinct elections but $X$ loses in the combined election formed by merging the ballots from the two  elections. For three-candidate IRV elections we provide necessary and sufficient conditions under which there exists a partition of the ballot set into two sets of ballots such that a given losing candidate wins each of the sub-elections. Applying these conditions, we use Monte Carlo simulations to estimate the frequency with which such partitions exist under various models of voter behavior.  We also analyze the frequency with which the paradox occurs in a large dataset of real-world ranked-choice elections to provide empirical probabilities. Our general finding is that IRV is highly susceptible to this paradox in three-candidate elections.
\end{abstract}


\maketitle

\section{Introduction}

The voting method of instant runoff voting (IRV), also referred to as the alternative vote, the plurality elimination rule, and colloquially as ``ranked-choice voting'' in the United States, is famously susceptible to many classical voting paradoxes \cite{FB83,N99,S95}. In this article we analyze the susceptibility of IRV to a lesser-studied paradox known as a \emph{reinforcement paradox}, which occurs when candidate $X$ wins under IRV in two distinct elections but $X$ loses in the combined election formed by merging the ballots from the two  elections. For three-candidate IRV elections we provide necessary and sufficient conditions under which there exists a partition of the ballot set into two subsets of ballots such that a given losing candidate wins each of the sub-elections. Applying these conditions, we use Monte Carlo simulations to estimate the frequency with which such partitions exist under various models of voter behavior.  We then analyze the frequency of the paradox in a large dataset of real-world ranked-choice elections to provide empirical probabilities. By characterizing the three-candidate IRV elections that demonstrate susceptibility to the reinforcement paradox and estimating its frequency, we view our work as a contribution to the larger agenda laid out by Felsenthal in \cite{F19}.

Reinforcement paradoxes have received attention in the social choice literature due to their interesting and surprising nature, much like other paradoxes to which IRV is susceptible, such as upward or downward monotonicity paradoxes. However, unlike monotonicity paradoxes, IRV's susceptibility to reinforcement paradoxes has direct administrative implications: to determine the IRV winner, it is not sufficient to know the winner (either plurality or IRV) at each precinct.   Instead, complete information about the numbers of ballots with each different ranking must be  transmitted—either electronically or physically—to a central counting location and collated before the IRV winner can be determined. By contrast, a  voting method  like plurality (or a scoring rule) allows each precinct to simply report the number of first-place  (and possible second-place, etc.) votes for each candidate, and for the winner to be determined by summing these results across precincts. As a result, IRV ballot tabulation and the subsequent determination of the winner is of necessity more complicated  than with methods like plurality. Thus, a study of reinforcement paradoxes is motivated by practical as well as theoretical considerations.

Our work in this paper connects to three previous strands of the literature concerning IRV and voting paradoxes. The first gives conditions on a preference profile which characterize when an election is susceptible to a given paradox. Examples are  Lepelley et al. \cite{LCB96}, Miller \cite{Mi17}, and Ornstein and Norman \cite{ON14}, which characterize the three-candidate profiles which demonstrate an upward or downward monotonicity paradox under IRV. Similarly, Graham-Squire \cite{GS24} characterizes three-candidate profiles which demonstrate a no-show paradox under IRV, including profiles which contain partial ballots. Kamwa et al. \cite{KMT23} study a variety of voting rules, but in particular provide conditions under which three-candidate profiles produce paradoxes such as negative participation paradoxes under IRV.

The second strand of literature typically uses the conditions given by the works cited above to estimate the frequency with which a given paradox occurs under IRV. For example, Ornstein and Norman \cite{ON14} estimate the frequency of monotonicity failures in IRV elections when the ballot data is generated using a spatial model. Similarly, Lepelley et al. \cite{LCB96} and Miller \cite{Mi17} estimate the frequency with which IRV produces monotonicity paradoxes under various theoretical models of voter behavior. Lepelley and Merlin \cite{LM01} analyze the frequency with which IRV demonstrates a no-show paradox in three-candidate elections under the impartial culture and impartial anonymous culture models. Kamwa et al. \cite{KMT23} analyze the probability that IRV (among other methods) demonstrates no-show or negative participation paradoxes on single-peaked domains. McCune and Wilson \cite{MW24} examine the susceptibility of IRV to negative participation paradoxes. 

More broadly, there is a large literature examining the extent to which various voting rules are susceptible to strategic manipulation. Such work need not rely on algebraic characterizations involving the vote totals in the preference profile; see \cite{AK99,FL06,GA,K93} for such examples.

The third strand is the expanding empirical social choice literature on IRV, enabled by the increasing availability of ranked-choice ballot data from jurisdictions such as Australia \cite{S24}, the United States \cite{O22}, and Scotland \cite{MGS24}.  Some of this literature concerns the empirical probability that IRV produces a particular kind of paradox. Graham-Squire and Zayatz \cite{GSZ20}, Graham-Squire and McCune \cite{GSM23}, and McCune and Graham-Squire \cite{MGS24} estimate the frequency with which real-world IRV elections demonstrate monotonicity and no-show paradoxes. Other empirical work includes case studies of particularly interesting elections which demonstrate a range of different paradoxes; see \cite{LS09, MM23}, for example. Beyond work on paradoxes, empirical research on IRV has analyzed rates of ballot errors in IRV elections \cite{C24, PR23}, the extent of ballot exhaustion \cite{BK15, GSM23}, the rate at which IRV exhibits the so-called ``spoiler effect'' \cite{MW23}, the extent to which ballot truncation can affect the IRV winner \cite{DMM, KGF, TUK}, and the frequency with which IRV selects the Condorcet winner \cite{GSM23,S23, S24}. There is also a growing, related empirical literature concerning the method of the single transferable vote (STV), a multiwinner version of IRV which is designed to achieve proportional representation given a set of preference ballots. This literature also investigates paradoxical and Condorcet-related outcomes \cite{MGS24, MMLS24}, as well as STV's propensity to achieve proportional representation \cite{BDDW24,BBMP24, M23}.

Of course, there is a much broader literature concerning paradoxes and voting methods other than IRV, and we do not attempt to survey that literature here. See \cite{D14, M88, P01} for an examination of the susceptibility of Condorcet consistent methods to various types of no-show paradox, for example.

Previous work on the reinforcement paradox is fairly limited. The notion of a reinforcement paradox dates to the work of Young \cite{Y74, Y75}, who defined the concept and  proved essentially that the only reasonable voting methods not susceptible to the paradox are positional scoring rules. Felsenthal and Nurmi \cite{FN19} show that various Condorcet inconsistent voting rules (including IRV) are susceptible to the paradox even in ``tame'' profiles in which a Condorcet winner exists and the voting rule chooses this candidate. Brandt et al. \cite{BMS22} provide some ``minimal'' examples of the reinforcement paradox under IRV where the number of candidates and voters is small. Brandt et al. \cite{BDP24} analyze the frequency of reinforcement paradoxes under various Condorcet consistent voting rules. We are aware of only two previous papers which attempt to measure the frequency with which IRV produces reinforcement paradoxes, both of which focus on the three-candidate case. Plassmann and Tideman \cite{PT14} measure this frequency under a spatial model, and Courtin et al. \cite{CMMS10} study the paradox under the impartial culture and impartial anonymous culture models of voter behavior. Our work differs from these studies in significant aspects which we explicate more fully below. In brief, the approach taken by these studies is to generate two elections and then say that the reinforcement paradox has occurred if the two elections have the same winner, but when the ballots from the two elections are combined, the winner is different. By contrast, we consider only one election and say that the paradox has occurred in this election if there exists a way to partition the ballots into two sub-elections where the winners of the sub-elections are the same but differ from the winner of the original election. Our work is complementary to theirs and allows us to focus  on one set of ballots and the characteristics required for it to demonstrate a susceptibility to this kind of inconsistency.

Another reason for our interest in this approach is that it parallels a well-established line of research regarding self-consistency in   allocation rules. For instance, Aumann and Maschler \cite{AM85} identify  \emph{consistency} (in which the outcome of  dividing assets among a set of agents doesn't change when applied to a subset of agents) as a crucial  characteristic of claims methods.  Balinski and Young \cite{BY01} define an analogous notion of  \emph{uniformity} to characterize divisor methods of apportionment. Converse consistency and other related ideas also appear in the cooperative game and bargaining literatures; see \cite{T11} for a survey.

Consistency in allocation rules is inherently different from consistency in voting rules:  consistency in allocation rules requires that if an allocation is recalculated after  a subset of agents ``departs'' with their share, the allocation does not change. Put another way, if the agents are partitioned into two subsets and their shared goods are pooled and then reallocated, the results are unchanged. 
We can adapt this latter interpretation to the social choice framework: if an election is partitioned into two subsets and the winner is determined in each of the subsets, the results should reflect the original winner.

The paper is structured as follows. Section \ref{prelims} gives preliminaries, including formal definitions of IRV and the reinforcement paradox. In Section \ref{section:conditions} we provide necessary and sufficient conditions for a three-candidate election to demonstrate a reinforcement paradox and discuss some of the consequences of these conditions. Proposition \ref{complete_cond} gives these conditions assuming all voters provide a complete ranking of the candidates while Proposition \ref{partial_cond} gives conditions when some voters cast partial ballots. Section \ref{threshold_results} considers the minimum percentage a losing candidate can earn and yet the election demonstrates a reinforcement paradox. Section \ref{models_section} estimates the probability that an election demonstrates a paradox under various theoretical models of voter behavior. Section \ref{empirical} analyzes real-world IRV elections from the United States and Scotland. This section provides empirical probabilities using a large dataset of elections, and provides a probabilistic analysis using non-parametric bootstrapping. We also provide results concerning geographic partitions of the electorate for the elections for which we have geographic information. Section \ref{section:discussion} gives a discussion of our results, and Section \ref{conclusion} concludes.

\section{Preliminaries}\label{prelims}

In this section we define IRV and different types of reinforcement paradoxes, illustrated with examples.

 In an election which uses IRV, voters cast preference ballots with a (possibly partial) linear ranking of the candidates. After an election, the ballots are aggregated into a preference profile, which shows the number of each type of ballot cast. For example, Table \ref{District_E} shows a preference profile involving the three candidates $C$, $G$, and $H$. The number 2211 denotes that 2211 voters ranked $G$ first, $H$ second, and $C$ third. In the real-world IRV elections we analyze in Section \ref{empirical}, voters are allowed to provide partial rankings. Table \ref{District_E} has three columns where voters do not provide a complete ranking--2003 voters rank only $G$ on their ballots, for example.

An IRV election proceeds in rounds where, in each round, the number of first-place votes for each candidate is calculated. If a candidate has a majority of the (remaining) first-place votes, they are declared the winner. If no candidate receives a majority of votes, the candidate with the fewest first-place votes is eliminated and the names of the remaining candidates on the affected ballots are moved up. The process repeats until a candidate is declared the winner. In the case when partial ballots are allowed, any ballot in which all candidates have been eliminated is considered ``exhausted'' and plays no further role in the determination of the winner.\footnote{In theory there are many ways we could treat exhausted ballots besides removing them from the election. However, as far as we are aware, any actual election office treats exhausted ballots in this fashion.}

When there are only three candidates, at most two rounds are required to determine a winner; in general, several rounds may be required. We focus only on the three-candidate case and thus do not consider more than two rounds of ballot counting. The following example demonstrates the IRV algorithm using the 2022 Alaska Senate District E election in Alaska. 

\begin{table} 
\begin{tabular}{l|ccccccccc}
Num. Voters & 2211 & 1439 &2003 & 2605& 241 &2689&2230&418&2873\\
\hline
1st choice & $G$ & $G$&$G$&$H$&$H$&$H$&$C$&$C$&$C$\\
2nd choice & $H$ & $C$& & $G$ & $C$& & $G$ & $H$\\
3rd choice & $C$ & $H$& & $C$ & $G$&&$H$ & $G$ \\
\end{tabular}

\caption{Vote totals in the 2022 Senate District E election in Alaska.}
\label{District_E}
\end{table}

\begin{example}
(2022 Alaska Senate District E.) This election contained the three candidates Roselynn Cacy (C), Cathy Giessel (G), and Roger Holland (H). The preference profile is displayed in Table \ref{District_E}.\footnote{Our vote totals differ from the officially reported vote totals by up to three votes. This minor discrepancy could be due to a recount which occurred after the election, or because the Alaska Division of Elections has not been completely clear about how it processes ``edge case'' ballots. Our analysis is unaffected.}

Initially, the vote totals for Cacy, Giessel, and Holland are 5521, 5653, and 5535, respectively. Cacy has the fewest votes and is eliminated from the election. As a result, 2230 votes are transferred to Giessel and 418 to Holland (the 2873 ballots ranking only Cacy are transferred to no one), and Giessel wins in the last round 7883 votes to 5953.
\end{example}

While Giessel is the winner, this election has the strange property that Cacy can point to two disjoint groups of voters such that Cacy would win under IRV if we use only the ballots from each group. Similarly, Holland can also break the electorate into two disjoint groups such that Holland wins when using the ballots from each group separately. We illustrate this paradoxical behavior in the next example.

\begin{example}\label{first_paradox_example}
    In the 2022 Senate District E election in Alaska, suppose we partition the ballots into two ballot sets $\mathcal{B}_1$ and $\mathcal{B}_2$ as shown for Partition 1 in Table \ref{District_E_partitions}. If we use just the ballots from $\mathcal{B}_1$, $H$ wins under IRV, as $C$ is eliminated first and enough votes are transferred to $H$ for him to win against $G$ in the final round. Similarly, if we use only the ballots from $\mathcal{B}_2$ then $H$ also wins under IRV, as in this case $G$ is eliminated first and enough votes are transferred to $H$ for him to win against $C$ in the final round. Thus, Roger Holland is able to point to two disjoint sub-electorates in which he would win the election using just the ballots from each sub-electorate, yet when we merge the ballots from the sub-electorates to recover the original ballot data, Holland loses.

    Table \ref{District_E_partitions} provides a second partition of the ballots, Partition 2, in which $C$ wins under IRV when using only the ballots from $\mathcal{B}_1$ and only the ballots from $\mathcal{B}_2$, yet when we merge the ballots to recover the original election data, $C$ loses. For both Cacy and Holland, in this election win $+$ win $=$ loss in some sense.
\end{example}

\begin{table}
\centering

\begin{tabular}{ll|ccccccccc}
\multicolumn{2}{c|}{Num. Voters} 
& 2230 & 418 & 2873 & 2211 & 1439 & 2003 & 2605 & 241 & 2689 \\
\hline
\multicolumn{2}{c|}{1st choice}  
& $C$ & $C$ & $C$ & $G$ & $G$ & $G$ & $H$ & $H$ & $H$ \\
\multicolumn{2}{c|}{2nd choice}  
& $G$ & $H$ &  & $H$ & $C$ &  & $G$ & $C$ &  \\
\multicolumn{2}{c|}{3rd choice}  
& $H$ & $G$ &  & $C$ & $H$ &  & $C$ & $G$ &  \\
\hline
\hline
\multirow{2}{*}{Partition 1} & $\mathcal{B}_1$ 
& 111 & 233 & 142 & 355 & 137 & 132 & 146 & 119 & 239 \\
& $\mathcal{B}_2$ 
& 2119 & 185 & 2731 & 1856 & 1302 & 1871 & 2459 & 122 & 2450 \\
\hline
\multirow{2}{*}{Partition 2} & $\mathcal{B}_1$ 
& 60 & 90 & 1900 & 200 & 1300 & 500 & 2580 & 11 & 59 \\
& $\mathcal{B}_2$ 
& 2170 & 328 & 973 & 2011 & 139 & 1503 & 25 & 230 & 2630 \\
\end{tabular}

\caption{Two partitions of the ballots in the 2022 Senate District E election in Alaska with no zero entries. In Partition 1, candidate $H$ wins both sub-elections; in Partition 2, candidate $C$ wins both sub-elections; yet candidate $G$ wins the overall election.}
\label{District_E_partitions}
\end{table}

Example \ref{first_paradox_example} demonstrates the notion of a \emph{reinforcement paradox}, the subject of this article. The study of this paradox dates to the work of Young \cite{Y74}. Young defined a voting method as a function $f$ which takes as input a set of ranked ballots and which outputs an election winner (or a set of winners, if there is some sort of tie).  We use a simplified version of Young's definition which does not address the issue of tied outcomes because ties almost never arise in the elections we study. In particular, none of our real-world elections have a tied outcome in any round. If $\mathcal{B}$ is a set of ballots  and $f(\mathcal{B})$ is the winner under voting method $f$, then Young defined the following notion of \emph{consistency} \cite{Y74}.

\begin{definition}\label{definition_consistency}
Let $\mathcal{B}_1$ and $\mathcal{B}_2$ be two sets of ballots over the same candidate set. A voting method $f$ is said to be \emph{consistent} if whenever $f(\mathcal{B}_1)=f(\mathcal{B}_2)$ then $f(\mathcal{B}_1\cup\mathcal{B}_2)=f(\mathcal{B}_1)$.
\end{definition}

A voting method that is not consistent is said to be susceptible to the \emph{reinforcement paradox}. This paradox has other names, such as the ``multiple districts paradox'' \cite{FB83}, and is a social choice analogue of the classical Simpson's Paradox from statistics. In \cite{Y75}, Young showed that consistency and a few other uncontroversial axioms completely characterize positional scoring rules. Thus, any ``reasonable'' voting method which is not a scoring rule is not consistent. We analyze only IRV and the function $f$  represents that voting method throughout the article.

To investigate the frequency with which reinforcement paradoxes occur, we must first decide what constitutes an ``occurrence'' or an ``instance'' of the paradox. There seem to be two viable approaches. The first, used by Courtin et al. \cite{CMMS10} and Plassman and Tideman \cite{PT14}, focuses on the ballot sets $\mathcal{B}_1$ and $\mathcal{B}_2$. Their work generates  ballot sets $\mathcal{B}_1$ and $\mathcal{B}_2$ independently using some probability distribution, and the paradox is said to occur if the election winner using $\mathcal{B}_1$ is the election winner using $\mathcal{B}_2$, but we obtain a different winner for the set $\mathcal{B}_1\cup \mathcal{B}_2$. A second approach, which is the one taken in this paper, is to start with a ballot set $\mathcal{B}$ and determine if it is possible to  decompose into two subsets $\mathcal{B}_1$ and $\mathcal{B}_2$ where the winners of $\mathcal{B}_1$ and $\mathcal{B}_2$ match but differ from the winner of the original election. That is,  we  select a single ballot set $\mathcal{B}$ using some probability distribution (or drawing from a set of real-world elections), and say this single election is susceptible to the reinforcement paradox if  $\mathcal{B}$ can be partitioned into two sub-elections $\mathcal{B}_1$ and $\mathcal{B}_2$ such that  $f(\mathcal{B}_1) = f(\mathcal{B}_2)$ but that the winner of the two sub-elections is different from $f(\mathcal{B}).$ 

In this analysis, we make no restrictions on the partitions. Indeed, our methodology involves finding necessary and sufficient conditions for such a partition to exist, without explicitly manufacturing one. Thus, if an election is susceptible to a reinforcement paradox in this sense, we do not  know \emph{a priori} anything about $\mathcal{B}_1$ or $\mathcal{B}_2$. The sub-elections may look ``similar'' to the original election, or look highly manufactured; there could be many possible partitions or few; the partitions may be more or less ``sensitive'' to small ballot perturbations. 

A consequence of this approach is that there may be many elections (either simulated or real-world) which are identified as susceptible to the paradox in which the ``paradoxical harm'' felt by the losing candidate who wins both sub-elections is highly theoretical. This is inevitable in the case of simulated elections in which the  models include no information about the voters (demographic, geographical, etc.) aside from their candidate preferences. It is also natural in the case of actual election data where preference profiles frequently contain no geographical information. 
Of course, instances of the paradox corresponding to partitions in which $\mathcal{B}_1$ and $\mathcal{B}_2$ represent some kind of meaningful split of the electorate may be regarded as more serious--the losing candidate who wins in these sub-elections might feel a more coherent sense of harm. 
We discuss this further at the end of the section.

Overall, we view our approach as complementary to that in Courtin et al. \cite{CMMS10} and Plassman and Tideman \cite{PT14}. We adopt our methodology  for the following reasons. First, we are motivated by a study of real-world elections (see Section \ref{empirical}). Given a single actual election, if we focus on decomposing one ballot set into a partition then we can meaningfully ask ``Did the 2022 Alaska Special House Election demonstrate susceptibility to a reinforcement paradox?''. An answer of Yes then means that a partition of the ballot set exists.  For real-world elections in which only preference profile information is available, aruably the most natural way to search for reinforcement failures is to consider all possible partitions of the electorate.

Second, this approach is in line with notions of ``occurrence'' for other paradoxes such as monotonicity paradoxes. For example, informally an election is said to demonstrate an upward monotonicity paradox if there exists a set of voters who could turn the winning candidate into a loser by shifting the winner to higher rankings on their ballots \cite{GSZ20, GSM23}. That is, given a single election we demonstrate this paradox by finding a subset of ballots which, when suitably changed, produces a paradoxical outcome. Our approach with reinforcement paradoxes has a similar flavor: while we do not change the ballots, we hunt for subsets of ballots which produce a paradoxical outcome. 

Third, much of the interesting work for the  approach of Courtin et al.  and Plassman and Tideman  has already been covered in \cite{CMMS10,PT14}, but the second approach has not been well-studied. Brandt et al. \cite{BMS22} give results about ``minimal'' reinforcement paradoxes arguably using the second approach; for example, for a fixed number of candidates and for various voting methods (including IRV), they give a profile with a minimal number of voters such that the preference profile can be partitioned into two ballots sets which demonstrate the paradox. But otherwise we are unaware of previous results about partitioning a single ballot set.
There is much interesting mathematics to explore in the second approach, such as providing necessary and sufficient conditions for an election to admit an appropriate partition of the ballot set. (We do exactly this, for three candidates, in   Section \ref{section:conditions}.)
For these reasons, we adopt the following definition for an election to demonstrate a reinforcement paradox.

\begin{definition}
An election with ballot set $\mathcal{B}$ is susceptible to (or ``demonstrates'') a \textbf{reinforcement paradox} under a voting method $f$ if there exists a partition of the ballot set $\mathcal{B}=\mathcal{B}_1 \cup \mathcal{B}_2$ such that $f(\mathcal{B}_1)=f(\mathcal{B}_2)$ and $f(\mathcal{B}_1)\neq f(\mathcal{B}_1\cup\mathcal{B}_2)$.
\end{definition}

Under IRV, if an election is susceptible to  a reinforcement paradox then the paradox's dynamics are similar to that of Example \ref{first_paradox_example}, in the sense that the elimination order in the two sub-elections are different. It is straightforward to show that if the IRV election using ballots from $\mathcal{B}_1$ has the same candidate elimination order as the IRV election using ballots $\mathcal{B}_2$ then the winner when using $\mathcal{B}_1\cup\mathcal{B}_2$ is necessarily the winner of the two sub-elections. Thus, when searching for a partition to demonstrate the paradox with three candidates, we look for sub-elections in which the two candidates in the final round of the election using $\mathcal{B}_1$ are not the same as the two candidates in the final round of the election using $\mathcal{B}_2$.

When there are three candidates, an election can demonstrate a reinforcement paradox in two different ways, leading to the two definitions below. We  refer to the candidate with the smallest number of first-place votes as the \emph{plurality loser}. This candidate can never be the winner in an IRV election since they are eliminated in the first round. We refer to the remaining losing candidate,  (who may receive either the largest or second-largest number of first-place votes) as the \emph{runner-up}.  

\begin{definition}
An election is susceptible to a \textbf{plurality loser reinforcement paradox}
if there exists a partition of the ballots $\mathcal{B}=\mathcal{B}_1\cup \mathcal{B}_2$ such that the plurality loser is the IRV winner when using only the ballots in $\mathcal{B}_1$ and when using only the ballots in $\mathcal{B}_2$.
\end{definition}

\begin{definition}
An election is susceptible to a \textbf{runner-up reinforcement paradox}
if there exists a partition of the ballots $\mathcal{B}=\mathcal{B}_1\cup \mathcal{B}_2$ such that the runner-up is the IRV winner when using only the ballots in $\mathcal{B}_1$ and when using only the ballots in $\mathcal{B}_2$.
\end{definition}

If an election is susceptible to both a plurality loser and runner-up paradox, we say the election is susceptible to a \textbf{double reinforcement paradox}.
The next section considers the problem of determining if a ballot set $\mathcal{B}$ in a three-candidate election can be partitioned to produce a runner-up or plurality loser reinforcement paradox.

To conclude this section, we elaborate further on what factors might make  some instances of the reinforcement paradox  arguably more serious than others. When motivating the study of reinforcement paradoxes, the classic story is as follows. \emph{Suppose candidate $A$ wins when using ballots only from the northern part of a city and wins using only ballots from the southern part, but $A$ loses when we combine the ballots. Does this not seem paradoxical or unfair in some way?} In this story the partition of the ballot set is ``natural'' in the sense that voters associated with ballots from $\mathcal{B}_1$ have some salient trait in common (in this case residence in the north), as do voters associated with ballots from $\mathcal{B}_2$. If voters can be partitioned to achieve the paradox via geography or other variables such as education level or race, then such a paradox is  more serious than if voters in these  partition subsets have nothing in common. The candidate who would have won in the sub-elections might think that having won both the north and the south, or both college-educated and non college-educated voters, for example, they should have won the election. In  Example \ref{first_paradox_example}, by contrast, we do not know if there is a common variable uniting voters in $\mathcal{B}_1$ or $\mathcal{B}_2$ in Partition 1 or 2.  Thus, the paradox demonstrated in that example is  of less concern than if we could definitively partition the voters by geography, race, education level, etc.  Because we consider all possible partitions, many of our demonstrated demonstrated paradoxes will  be of less concern in this sense.

The primary limitation of our study is that in real-world ranked-choice elections, we often have minimal information about the electorate---usually  just a preference profile. In such cases, we cannot know whether a given partition represents something meaningful like a north-south split. However, for a subset of  real-world elections, additional geographic information is available, such as the location of voters by precinct.  For these elections,  can search for natural geographic partitions of the electorate. The results, including several interesting examples, are given in Section \ref{section:precincts}.

\section{Necessary and Sufficient Conditions}\label{section:conditions}

To determine the  likelihood that a three-candidate election is susceptible to a reinforcement paradox, we must be able to  characterize when it is possible to  partition $\mathcal{B}$ into  two sub-elections $\mathcal{B}_1$  and $\mathcal{B}_2$ such that either the runner-up or the plurality loser wins both sub-elections.

Assume that the preference profiles for complete and partial rankings are  as shown in Tables \ref{profile_complete_ballots} and \ref{profile_partial_ballots} where in each case the expressions $X_i, Y_i$ and $Z_i$ indicate the number of voters with the corresponding ranking. Note that in Table  \ref{profile_partial_ballots} we do not list ballots of length two, as such ballots convey the same ranking information as the corresponding complete ballots. 

\begin{table}[tbh]
\centering
\begin{tabular}{l|cccccc}
Num. Voters & $X_1$ & $X_2$  & $Y_1$ & $Y_2$ &  $Z_1$ & $Z_2$ \\
\hline
1st choice & $A$  &$A$ & $B$ & $B$ & $C$ & $C$ \\
2nd choice &$B$  &$C$ & $A$ & $C$ & $A$& $B$ \\
3rd choice &$C$  &$B$ & $C$ & $A$ & $B$ & $A$ \\
 
\end{tabular}
\caption{Preference profile with complete preferences.}
\label{profile_complete_ballots}

\end{table}

\begin{table}[tbh]

\centering
\begin{tabular}{l|ccccccccc}
Num. Voters & $X_1$ & $X_2$ & $X_3$ & $Y_1$ & $Y_2$ & $Y_3$ & $Z_1$& $Z_2$ & $Z_3$ \\
\hline
1st choice & $A$& $A$ &$A$ &$B$ & $B$ & $B$ &$C$&$C$&  $C$ \\
2nd choice &  $B$&$C$ && $A$ & $C$ &&$A$  &$B$ \\
3rd choice & $C$ &$B$  && $C$ & $A$ &&$B$&  $A$ \\
\end{tabular}
\caption{Preference profile with some partial preferences.}
\label{profile_partial_ballots}
\end{table}

In the complete ballot case, let $X=X_1+X_2$, $Y=Y_1+Y_2$, and $Z=Z_1+Z_2$; in the partial ballot case, let $X'=X_1+X_2+X_3$, $Y'=Y_1+Y_2+Y_3$, and $Z'=Z_1+Z_2+Z_3$.
We assume without loss of generality that $A$ is the IRV winner and $C$ is the plurality loser. 

First, we observe that if $A$ has a  majority of first-place votes in the first round then the election is not susceptible to any kind of reinforcement paradox. That is, we cannot observe a paradox if $X>Y+Z$ (in the complete ballot case) or $X^\prime>Y^\prime+Z^\prime$ (in the partial ballot case).  The reason is that if a $A$ has a majority, then they must earn a majority of first-place votes in either $\mathcal{B}_1$ or $\mathcal{B}_2$ for any partition of the ballot set.

\begin{observation}
If an election contains a majority candidate then the election is not susceptible to a reinforcement paradox.
\end{observation}

Thus, since $A$ is the IRV winner, we assume in the complete ballot case that 

\[Z<X, \ Z<Y  \text{ and } Y+Z_2<X+Z_1\]

and in  the partial ballot case that

\[Z'<X', \ Z'<Y' \text{ and } Y'+Z_2<X'+Z_1.\]

In addition, as remarked previously, if there is a partition such that either the runner-up $B$ or the plurality loser $C$ wins both sub-elections in two rounds, then the order of elimination in $\mathcal{B}_1$ must be different from that in $\mathcal{B}_2$.  (Suppose, for instance, that $B$ wins both sub-elections in two rounds and that $C$  is the plurality loser in both sub-elections. Using the notation of Table \ref{complete_cond}, this implies $X^\prime+Z_1< Y^\prime+Z_2$ in both sub-elections, which implies the same inequality holds in the original election.)  This leads to the following. 

\begin{observation}
If there is a partition of the  electorate in which either the runner-up $B$ or the plurality loser $C$ wins both sub-elections in the second round of counting, then the plurality loser is different in each sub-election.
\end{observation}

As the following Propositions indicate, it is also possible for an election to have a partition in which  the runner-up wins a majority of first-place votes in one sub-election while winning in round 2 in the other sub-election.  

\begin{proposition}\label{complete_cond}
(Complete Ballots). An election demonstrates a runner-up reinforcement paradox if and only if the election satisfies
\begin{equation}\label{RU_round2}
Y>X-Z_2+1 \text{ and } Y>-X_1+X_2+Z_1-Z_2+1.
\end{equation}
That is, conditions (\ref{RU_round2}) are necessary and sufficient for there to exist a partition of the election into two sub-elections in both of which the runner-up is the winner.  If in addition, 
\begin{equation}\label{RU_round1}
Y>X +1\text{ and } Y> -X_1+X_2+Z+1,
\end{equation}
then there is a partition in which the runner-up wins a majority of first-place votes in one of the sub-elections and wins in the second round in the other sub-election.

%


An election demonstrates a plurality loser reinforcement paradox if and only if  

\begin{equation}\label{PL_round2}
\begin{aligned}
Z &> X - Y_2+1, \\
Z &> X_1 - X_2 + Y_1 - Y_2+1, \\
Z &> -X_2 + Y + 1, \text{ and }\\
Z &> \tfrac{1}{3}(X + Y) + 1.
\end{aligned}
\end{equation}

In this case, there is a partition in which the plurality loser wins in the second round in both sub-elections.
\end{proposition}

\begin{proposition}\label{partial_cond}
(Partial Ballots). An election demonstrates a runner-up reinforcement paradox if and only if the election satisfies

\begin{equation}\label{part_round2}
\begin{aligned}
Y' &> X' - Z_2 + 1, \\
Y' &> -X_1 + X_2 + Z_1 - Z_2 + 1, \text{ and}\\
Y' &> \tfrac{1}{3}(-X_1 + 2X_2 + 2X_3 + Z_1 - Z_2) + 1.
\end{aligned}
\end{equation}

That is, conditions (\ref{part_round2}) are necessary and sufficient for there to exist a partition of the election into two sub-elections in both of which the runner-up is the winner. If in addition, 
\begin{equation}\label{part_round1}
Y^\prime>X^\prime+1 \text{ and } Y^\prime>-X_1+X_2+Z^\prime+1,
\end{equation}
then there is a partition in which the runner-up wins a majority of first-place votes in one of the sub-elections and wins in the second round in the other sub-election.

An election demonstrates a plurality loser reinforcement paradox if and only if  
\begin{equation}\label{part_PL_round2}
\begin{aligned}
Z' &> X' - Y_2 + 1, \\
Z' &> X_1 - X_2 + Y_1 - Y_2 + 1, \\
Z' &> -X_2 + Y' + 1, \text{ and}\\
Z' &> \tfrac{1}{3}(X' + Y') + 1.
\end{aligned}
\end{equation}

In this case, there is a partition in which the plurality loser wins in the second round in both sub-elections.
\end{proposition}

In subsequent sections we sometimes consider a ``limiting case'' where the number of voters approaches infinity. In this case, the values at the top of a preference profile can be considered proportions rather than vote counts, in which case the ``+1'' term in the above propositions disappears. That is, the above propositions still hold for the limiting case by eliminating the plus 1 term, which corresponds to dividing 1 by a larger and larger number of voters as we take the limit.

The conditions in Proposition \ref{partial_cond} reduce to those in Proposition \ref{complete_cond} when $X_3=Y_3=Z_3=0$. The  condition $Y^\prime>\frac{1}{3}(-X_1+2X_2+2X_3+Z_1-Z_2)+1$ in  Proposition \ref{partial_cond} reduces to $Y>\frac{1}{3}(-X_1+2X_2+Z_1-Z_2)+1$ and does not have an analogue in Proposition \ref{complete_cond} because it is implied by the other conditions. To see this,  note that

\begin{align*}
X_1-2X_2&+3Y-Z_1+Z_2-3=(X-Y+Z_1-Z_2)\\
&+2(-X+Y+Z_2-1)+(X_1-X_2+Y-Z_1+Z_2-1)+X_1+(Y-Z).\end{align*}
But $-X+Y+Z_2-1\ge 0$ and $X_1-X_2+Y-Z_1+Z_2-1\ge 0$ by the remaining two conditions in (\ref{RU_round2}), and $Y-X>0$ and $X-Y+Z_1-Z_2>0$ by assumption. Hence $X_1-2X_2+3Y-Z_1+Z_2-3>0$, which implies $Y>\frac{1}{3}(-X_1+2X_2+Z_1-Z_2)+1$.

Before sketching the proof of the propositions, we make a few remarks. First, an election can be susceptible to either a runner-up or plurality loser reinforcement paradox when the IRV and plurality winners in the original election are distinct or when they coincide, for the complete ballot or partial ballot case. 

Second, condition (\ref{RU_round1}) clearly can only be satisfied if  the IRV runner-up is the plurality winner. In this case, as the proof  will demonstrate, if there is a partition in which the runner-up wins a majority of first-place votes in one sub-election, they must win in two rounds in the other sub-election with the original IRV winner being eliminated in round 1. 

Third, it is easy to show that it is not possible to have a partition in which  the plurality loser has a majority of first-place votes in one sub-election and wins in two rounds in the other sub-election. Instead, a plurality loser reinforcement paradox can only occur when the plurality loser wins in two rounds in both sub-elections.

We prove a portion of Proposition  \ref{complete_cond}, as outlined below. The remaining proofs are analogous (and considerably longer).

\begin{proof} (of Proposition  \ref{complete_cond})
We prove that the inequalities in (\ref{RU_round1}) of Proposition  \ref{complete_cond} are necessary and sufficient for there to be a partition in which the runner-up wins a majority of first-place votes in one of the sub-elections and wins in the second round in the other sub-election. 

Suppose the original election has a voter profile as shown  in  Table \ref{profile_complete_ballots}. Since candidate $B$ is not  eliminated in either sub-election, we can collapse the two columns  in Table \ref{profile_complete_ballots} that rank $B$ first and simply consider the number of number of voters $Y$ who rank $B$ first and the remaining candidates in any order. Now suppose $\mathcal{B}_1$ and $\mathcal{B}_2$ are a partition of the  election into two sets of ballots. Let the voter profile of $\mathcal{B}_1$ be as in Table \ref{profile_complete_ballots} with $x_1$ replacing $X_1$, $x_2$ replacing $X_2$, $y$ replacing $Y$, $z_1$ replacing $Z_1$ and $z_2$ replacing $Z_2$, for some $0 \le x_1 \le X_1$, $0 \le x_2 \le X_2$,  $0 \le y \le Y,$ etc. Let $\mathcal{B}_2$ be the same except with $X_1-x_1,$ $X_2-x_2$, $Y-y$, $Z_1-z_1$ and $Z_2-z_2$.  

Suppose $B$ wins by a  majority of first-place votes in $\mathcal{B}_2$. Clearly, it also cannot win by a majority in  $\mathcal{B}_1$ (otherwise $B$ would have a majority of first-place votes in the original election). So $B$ must win in two rounds in $\mathcal{B}_1.$ We claim that $A$ must be the plurality loser in $\mathcal{B}_1.$

To see this, suppose $C$ is the plurality loser in $\mathcal{B}_1.$ After $C$ is eliminated,
$$x_1+x_1+z_1 <  y+z_2.$$
On the other hand, since $B$ wins in a majority in $\mathcal{B}_2,$
$$(X_1-x_1)+(X_2-x_2)+(Z_1-z_1)+(Z_2-z_2) < Y-y.$$

Solving both of these inequalities for $z_1$, and recalling that $z_1$ is an integer, we have 
$$X_1+X_2-Y+Z_1+Z_2-x_1-x_2+y-z_2 +1\le z_1 \le -x_1-x_2+y+z_2-1,$$ which implies 
$$X_1+X_2-Y+Z_1+Z_2+2 \le 2z_2.$$
Since $z_2 \le Z_2$, this implies $X_1+X_2-Y+Z_1+Z_2+2 \le 2Z_2$ or
$$X_1+X_2+Z_1+2 \le Y+Z_2,$$
which is a contradiction, since $A$ is the IRV winner in the original election. This proves the claim.

Thus $A$ must be the plurality loser in $\mathcal{B}_1$  with $B$ winning in round 2.  Let $z=z_1+z_2$ be the number of voters in  $\mathcal{B}_1$ who first-rank $C$ and let $Z-z=(Z_1-z_1)+(Z_2-z_2)$ be the number of analogous voters in  $\mathcal{B}_2$. We require (recalling that $X=X_1+X_2$)
\begin{align}
x_1+x_2& < y  \label{eq_1}\\
x_1+x_2& < z \label{eq_2} \\
x_2+z & < y+x_1   \label{eq_3}\\
 (X-x_1-x_2)+ (Z-z) &< (Y-y). \label{eq_4}
\end{align}
We simplify these inequalities in a series of steps in which we sequentially eliminate variables.  In the first step, we collect lower and upper bounds for $y$.

Step 1: Solving (\ref{eq_1}), (\ref{eq_2}) and (\ref{eq_4}) for $y$ and adding the boundary constraints $0 \le y \le Y,$ we have

$$\max\{  x_1+x_2 +1,  -x_1+x_2+z +1\} \le y \le \min\{Y, -X+Y-Z+x_1+x_2+z-1 \}$$

Since each lower bound must be bounded above by each upper bound, this implies 

\begin{align*}
x_1+x_2+1 & \le Y & \Rightarrow  & x_2   \le Y-x_1-1\\
x_1+x_2+1 & \le -X+Y-Z+x_1+x_2+z-1 & \Rightarrow  & z \ge X-Y+Z+2 \\
-x_1+x_2+z +1 & \le Y & \Rightarrow  & z   \le Y+x_1-x_2-1\\
-x_1+x_2+z +1 & \le -X+Y-Z+x_1+x_2+z-1 & \Rightarrow  & x_1 \ge \frac{1}{2}(X-Y+Z+2). \\
\end{align*}

Step 2: Collecting all the bounds (including  (\ref{eq_2})) on $z$, along with the boundary constraint $0 \le z \le Z$, we have

$$\max\{  x_1+x_2+1,  X-Y+Z+2 \} \le z \le \min\{Z, Y+x_1-x_2-1 \}.$$

Again, since each lower bound must be bounded above by each upper bound, we have
\begin{align*}
x_1+x_2+1 & \le Z & \Rightarrow  & x_2   \le Z-x_1-1\\
x_1+x_2+1 & \le Y+x_1-x_2-1 & \Rightarrow  & x_2 \le \frac{1}{2}(Y-2) \\
 X-Y+Z+2 & \le Z & \Rightarrow  & Y \ge X+2\\
 X-Y+Z+2 & \le Y+x_1-x_2-1 & \Rightarrow  & x_2 \le -X+2Y-Z+x_1-3 \\
\end{align*}

Step 3:  Collecting all the bounds  on $x_2$, we have
$$\max\{ 0 \} \le x_2 \le \min\{X_2, Y-x_1-1, Z-x_1-1, \frac{1}{2}(Y-2), -X+2Y-Z+x_1-3  \}.$$
Hence
\begin{align*}
0 & \le Y-x_1-1 & \Rightarrow  & x_1   \le Y-1\\
0 & \le Z-x_1-1 & \Rightarrow  & x_1 \le Z-1 \\
0 & \le \frac{1}{2}(Y-2) & \Rightarrow  & Y \ge 2\\
0 & \le -X+2Y-Z+x_1-3 & \Rightarrow  & x_1 \ge X-2Y+Z+3 \\
\end{align*}

Step 4: Finally, collecting all the bounds on  $x_1$, we have

$$\max\{ 0, \frac{1}{2}(X-Y+Z+2),  X-2Y+Z+3  \} \le x_1 \le \min\{X_1, Z-1,   \}.$$
This implies
\begin{align*}
0 & \le Z-1 & \Rightarrow  & Z \ge 1\\
\frac{1}{2}(X-Y+Z+2) & \le X_1 & \Rightarrow  & Y \ge -X_1+X_2+Z+2 \\
\frac{1}{2}(X-Y+Z+2) & \le Z-1 & \Rightarrow  & Y \ge X-Z+4 \\
X-2Y+Z+3& \le X_1 & \Rightarrow  & Y \ge \frac{1}{2}(X+Z+3) \\
X-2Y+Z+3 & \le Z-1 & \Rightarrow  & Y \ge \frac{1}{2}(X+4)\\
\end{align*}

Step 5: Collecting all the inequalities relating $X_1$, $X_2$, $Y$ and $Z$, these reduce to (assuming $Z \ge 4$): 
$$Y \ge X+2 \quad \text{and} \quad Y \ge -X_1+X_2+Z+2.$$ 
Since $Y$ is an integer, these conditions are equivalent to the conditions   $Y >X+1$ and $Y> -X_1+X_2+Z+1$   in (\ref{RU_round1}) of Proposition \ref{complete_cond}. \end{proof}

 Propositions \ref{complete_cond} and \ref{partial_cond} identify the necessary conditions on the original election for a runner-up or plurality loser reinforcement paradox to occur but do not identify any partitions. However, the inequalities  that arise at each step of the propositions' proofs can be used to identify  election partitions leading to the paradoxes.

For instance, from the proof above of the necessity of (\ref{RU_round1}) of Proposition \ref{complete_cond}, we have

\begin{enumerate}
\item $\max\{  x_1+x_2 +1,  -x_1+x_2+z +1\} \le y \le \min\{Y, -X+Y-Z+x_1+x_2+z-1 \}$
\item $\max\{  x_1+x_2+1,  X-Y+Z+2 \} \le z \le \min\{Z, Y+x_1-x_2-1 \}$
\item $\max\{ 0 \} \le x_2 \le \min\{X_2, Y-x_1-1, Z-x_1-1, \frac{1}{2}(Y-2), -X+2Y-Z+x_1-3  \}$
\item $\max\{ 0, \frac{1}{2}(X-Y+Z+2),  X-2Y+Z+3  \} \le x_1 \le \min\{X_1, Z-1,   \}.$
\end{enumerate}

Thus, if $X_1=50,$ $X_2=50$, $Y=102$, $Z_1=50$ and $Z_2=45,$ we can select  (among other possibilities), a partition in which the first sub-election is

$$x_1=48, x_2=40, y= 89 \mbox{ and } z = 95,$$ 
and the second partition is
$$X_1-x_1=2, X_2-x_2=10, Y-y=13, Z-z=0.$$
It is easily seen that $A$ wins the original election, but in $\mathcal B_1$ candidate $B$ wins in two rounds after $A$ is eliminated and in $\mathcal B_2$, candidate $B$ has a majority.

\subsection{Consequences of the Algebraic Conditions}

The inequalities identified in Propositions \ref{complete_cond} and \ref{partial_cond} allow us to identify further anomalies. For instance, 
Table \ref{CWparadox} demonstrates that it is possible for a candidate to be the IRV winner, a Condorcet winner and a plurality winner and yet the profile still  satisfies the  requirements in Propositions \ref{complete_cond} for both a runner-up and plurality loser paradox.

\begin{table}
\centering
\begin{tabular}{l|c|c|c|c|c|c}
Num. Voters & 1001 & 50 & 521& 480 & 561& 439\\
\hline
1st & $A$ & $A$ & $B$ & $B$& $C$ & $C$\\ 
2nd & $B$ & $C$ & $A$ & $C$& $A$ & $B$ \\
3rd & $C$ & $B$ & $C$&$A$ & $B$ & $A$\\
\end{tabular}
\caption{(Top) A profile with the IRV,  plurality,  and   Condorcet winners coinciding that demonstrates both  runner-up and plurality loser paradoxes. }
\label{CWparadox}
\end{table}

The next example shows that reinforcement paradoxes can be nested: an election may exhibit a runner-up or plurality loser paradox and each sub-election demonstrate an additional reinforcement paradoxes.

\begin{example}\label{example:nested}
 Consider the profile in Table \ref{nested1} where   $A$ is the IRV winner. The election demonstrates a runner-up paradox and $B$ wins in both sub-elections  $\mathcal{B}_1$ and $\mathcal{B}_2$. However each of these sub-elections themselves satisfies the conditions for a runner-up paradox. This means that  sub-election $\mathcal{B}_1$, in which $C$ is the plurality loser and $A$ is the runner-up, can be further subdivided into two smaller sub-elections in which $A$ is the IRV winner. (One such partition is $\mathcal{B}_{11}$  which corresponds to $X_1=12$, $X_2=31$, $Y_1=0$, $Y_2=25$, $Z_1=0$ and $Z_2=17$ and $\mathcal{B}_{12}$ which corresponds to $X_1=28$, $X_2=31$, $Y_1=5$,$Y_2=26$, $Z_1=1$ and $Z_2=31$.) And sub-election $\mathcal{B}_2$, in which $A$ is the plurality loser and $C$ is the runner-up, can be further subdivided into two smaller sub-elections in which $C$ is the IRV winner. (One such partition is $\mathcal{B}_{21}$  which corresponds to $X_1=50$, $X_2=0$, $Y_1=4$, $Y_2=25$, $Z_1=30$ and $Z_2=0$ and $\mathcal{B}_{22}$ which corresponds to $X_1=10$, $X_2=38$, $Y_1=41$, $Y_2=54$, $Z_1=69$ and $Z_2=2$.)

\end{example}

\begin{table}
\centering
\begin{tabular}{ll|cccccc}
\multicolumn{2}{c|}{Num. Voters} & 100 & 100 & 50 & 130 & 100 & 50  \\
\hline
\multicolumn{2}{c|}{1st choice}  & $A$ & $A$ & $B$ & $B$ & $C$ & $C$ \\
\multicolumn{2}{c|}{2nd choice}  & $B$ & $C$ &$A$ & $C$ & $A$ &$B$ \\
\multicolumn{2}{c|}{3rd choice}  & $C$ & $B$ &$C$ & $A$ & $B$ &$A$  \\
\hline
\hline
\multirow{2}{*}{Partition} & $\mathcal{B}_1$ & 40&62&5&51&1&48 \\
 & $\mathcal{B}_2$ & 60&38&45&79&99&2 \\
\end{tabular}
\caption{A profile that demonstrates a runner-up paradox with sub-elections that demonstrate additional reinforcement paradoxes. }
\label{nested1}
\end{table}

Additional instances of the paradox may occur in which an election can be divided into three or more sub-elections in which either the runner-up or plurality loser  wins.

\begin{table}
\centering
\begin{tabular}{l|c|c|c|c}
Num. Voters & 1000 & 999 & 560&438\\
\hline
1st choice& $A$ & $B$ & $C$ & $C$\\ 
2nd choice& $C$ & $C$ & $A$ & $B$ \\
3rd choice& $B$ & $A$ & $B$&$A$ \\
\end{tabular}

\caption{A profile that demonstrates a plurality loser, but not a runner-up, reinforcement paradox.}
\label{PL_not_RU_ex}
\end{table}

Finally, we note that elections may be susceptible to a plurality loser paradox but not a runner-up paradox, as demonstrated by the profile in  Table \ref{PL_not_RU_ex} where neither Condition (1) nor (2) of Proposition \ref{complete_cond} is satisfied, but the plurality loser inequalities are satisfied.  Such preference profiles can arise when voter preferences are ``single-peaked.'' In this case, candidate $C$ is like a ``centrist'' candidate in a spatial model, as if $A$, $B$, and $C$ have been arranged on a left-to-right line, with $C$ in the middle. 

Of course, real-world elections typically do not return ballot data resembling the profile in Table \ref{PL_not_RU_ex}. In real elections, the IRV runner-up tends to be  a much stronger candidate than the plurality loser. Thus, we expect that if a real-world election is susceptible to a plurality loser reinforcement paradox then it will also be susceptible to a runner-up paradox.  Our results in Section \ref{empirical} bear out this expectation.

\section{Threshold Results}\label{threshold_results}

In this section we examine the minimal percentage each candidate can earn while the  election still demonstrates a reinforcement paradox. For some classical paradoxes such as upward and downward monotonicity paradoxes, an election can exhibit such a paradox only if each of the three candidates is not ``weak.'' We show this is not the case for a general reinforcement paradox in Proposition \ref{small_amt_prop}, where the plurality loser can have an arbitrarily low percentage of the first-places votes and yet an election is susceptible to a (runner-up) reinforcement paradox. We show in Proposition \ref{second_prop} that for an election to be susceptible to a runner-up (respectively plurality loser) paradox, the IRV runner-up (respectively plurality loser) must earn at least 25\% of the first-place votes.

In this section, we assume that the number of voters approaches infinity, so that we use proportions of votes rather than integers. We also assume that each voter provides a complete ranking as this does not affect the logic of the proofs. Since the number of voters approaches infinity, we ignore the ``+1'' which appears in the Propositions in Section \ref{section:conditions}, as this term limits to zero.

\begin{proposition}\label{small_amt_prop}
For any $\epsilon>0$, there is an election which is susceptible to the reinforcement paradox and the first-place vote share for the plurality loser is less than $\epsilon$.
\end{proposition}

\begin{proof}

Consider the profile in Table \ref{first_prop_min}, where $p_A$ is the proportion of voters who cast the ballot $A\succ B \succ C$ and also the proportion of voters who cast the ballot $A\succ C \succ B$, and $\beta$ and $\delta$ are positive and very small. Furthermore, assume $2p_A+\beta<0.5$ and $4p_A+\beta+\delta=1$ (which implies $\delta>\beta)$. Note that $A$ is the IRV winner and $C$ is the plurality loser. 

This profile satisfies condition (2) of Proposition \ref{complete_cond} and thus the election demonstrates a runner-up reinforcement paradox. Since $\delta$ is arbitrary, we can make $\delta < \epsilon$ and the proposition is proven. For completeness, we include a partition of the electorate into two sub-electorates, as shown in the table. Let $\delta_B$ and $\delta_A$ be positive small parameters such that $0<\delta_B<\delta_A<\delta/2$. Then  $B$ wins when using the ballots of both sub-electorates separately. \end{proof}

\begin{table}
\centering
\begin{tabular}{ll|cccc}
\multicolumn{2}{c|}{Prop. Voters} & $p_A$ & $p_A$ & $2p_A+\beta$ & $\delta$  \\
\hline
\multicolumn{2}{c|}{1st choice}  & $A$ & $A$ & $B$ & $C$ \\
\multicolumn{2}{c|}{2nd choice}  & $B$ & $C$ &$A$ & $A$  \\
\multicolumn{2}{c|}{3rd choice}  & $C$ & $B$ &$C$  & $B$  \\
\hline
\hline
\multirow{2}{*}{Partition} & $\mathcal{B}_1$ & $\delta-\delta_A$&0&$\delta-\delta_B$&$\delta$\\
 & $\mathcal{B}_2$ & $p_A-\delta+\delta_A$&$p_A$&$2p_A+\beta-\delta+\delta_B$&0 \\

\end{tabular}

\caption{A profile in which the plurality loser earns a small proportion of first-place votes and the election demonstrates a (runner-up) reinforcement paradox. }
\label{first_prop_min}
\end{table}

Thus, for a large enough electorate the plurality loser can control an arbitrarily small amount of first-place votes and the election can demonstrate a reinforcement paradox. This proposition distinguishes reinforcement paradoxes from other classical paradoxes such as upward or downward monotonicity paradoxes. For a three-candidate election to demonstrate an upward (resp. downward) paradox, each candidate must earn at least 25\% (respectively $16\frac{2}{3}\%$) of the first-place votes \cite{Mi17}. That is, for some classical paradoxes none of the candidates can have an extremely weak first-place showing; each candidate must be ``viable'' in some sense. Proposition \ref{small_amt_prop} shows this is not the case for the reinforcement paradox. However, in the next proposition we show that for a runner-up (resp. plurality) reinforcement paradox, the runner-up (resp. plurality loser) must earn at least 25\% of the first-place votes.

\begin{proposition}\label{second_prop}

If an election is susceptible to a runner-up (resp. plurality loser) reinforcement paradox then the runner-up (resp. plurality loser) must earn at least 25\% of the first-place votes in round 1. This bound is tight in the limit as the number of voters approaches $\infty$.

\end{proposition}

\begin{proof}

\textbf{Runner-up case.} Recall that in order for an election to demonstrate a reinforcement paradox, no candidate can earn an initial majority, and thus candidate $A$ must earn less than 50\% of the first-place votes. This implies $B$ must earn at least $25\%$ of the first-place vote. 

To see that that bounds are tight,   consider the profile in Table \ref{prop_min} where   $\epsilon_A, \epsilon_B, \epsilon_1$ and $\epsilon_2\ge 0$ are  small parameters satisfying $ \epsilon_B=\epsilon_A-\epsilon_1+\epsilon_2.$ If   $\epsilon_B>\epsilon_2-\epsilon_A$ then this profile satisfies the conditions in (2) of Proposition \ref{complete_cond}.

\begin{table}
\centering
\begin{tabular}{l|c|c|c|c}
Num. Voters & $0.50-\epsilon_A$ & $0.25+\epsilon_B$ & $\epsilon_1$ & $0.25-\epsilon_2$\\
\hline
1st choice& $A$ & $B$ & $C$ & $C$\\
2nd choice& $B$ & $A$ & $A$ & $B$ \\
3rd choice& $C$ & $C$ & $B$ & $A$\\
\end{tabular}
\caption{Minimum profile for  runner-up reinforcement paradox}
\label{prop_min}
\end{table}

\textbf{Plurality loser case.} From Proposition \ref{complete_cond},  $Z>\frac{1}{3}(X+Y)$ and thus the plurality loser must have at least $25\%$ of the first-place vote. On the other hand, consider the profile in Table \ref{prop_min2} in which the plurality loser has just over 25\% of the vote and yet,  for some suitably small $\epsilon>0,$ satisfies the plurality loser inequalities from  Proposition \ref{complete_cond}. \end{proof}
\begin{table}
\centering
\begin{tabular}{l|c|c|c}
Num. Voters & $0.375-\epsilon$ & $0.375-2\epsilon$  & $0.25+3\epsilon$\\
\hline
1st choice& $A$ & $B$ & $C$ \\
2nd choice& $C$ & $C$ & $A$  \\
3rd choice& $B$ & $A$ & $B$ \\
\end{tabular}
\caption{Minimum profile for  plurality loser reinforcement paradox}
\label{prop_min2}
\end{table}


We conclude this section by noting that among elections susceptible to a reinforcement paradox there is
no lower bound on how small one of the sub-elections can be, as a percentage of the whole.

\begin{example}\label{small_subelection}

Let $m\in\mathbb{N}$ be a large even number. Consider an election $\mathcal{B}$ with  $3m-1$ ballots (complete or partial) in which the first and second preferences are distributed as follows:  $m$ ballots of the form $A\succ B$; $m+1$ ballots of the form $B\succ A$;  $m/2$ ballots of the form $C\succ A$; and $m/2-2$ ballots of  the form $C\succ B$. Candidate  $A$ wins this election in two rounds after the $C$ is eliminated in round 1.

Let $\mathcal{B}_1$ consist of the sub-election with $9$ ballots:  $2$ ballots of the form $A\succ B$; $3$ ballots  of the form $B\succ A$; and $4$ ballots of the form $C\succ A$. Let $\mathcal{B}_2$ consist of the sub-election with $3m-10$ ballots: $m-2$ ballots of the form $A\succ B$; $m-2$ ballots of the form $B\succ A$; $m/2-4$ ballots of the form $C\succ A$; and $m/2-2$ ballots of the form $C\succ B$. Then $\mathcal{B}_1$ and $\mathcal{B}_2$ form a partition of $\mathcal{B}$, and candidate $B$ wins both of these sub-elections. Since $m$ can be arbitrarily large, this implies there exists an election demonstrating a reinforcement paradox in which one of the sub-elections consists of an arbitrarily small fraction of the ballots in the original election.
\end{example} 

Example \ref{small_subelection} is interesting as it suggests that $B$ was ``close'' to being a winner in the original election. In fact, $B$ might feel it particularly unfair that the addition of a handful of ballots (from $\mathcal{B}_1$, where $B$ was  a winner) caused $B$ to go from winner (in $\mathcal{B}_2$) to loser (in $\mathcal{B}_1 \cup \mathcal{B}_2$).

\section{Results under various models}\label{models_section}

In this section we use different models of voter behavior to generate millions of synthetic elections. Among the many possible models we could use, we choose models which have been used extensively in prior social choice research. For each model, we generate 100,000 profiles and check whether a generated election is susceptible to a reinforcement paradox. Thus, our primary tool for providing (estimated) probabilities is Monte Carlo simulation, although we can provide exact probabilities using other means for the impartial anonymous culture (IAC)  model.

We are interested in probabilities for relatively large electorates, as the real-world elections we analyze in the next section have electorate sizes ranging from a few hundred to almost one million (although most of the elections contain 20,000 or fewer voters). Thus, our simulations use an electorate size of $V=1001$ voters. We  also ran simulations with $V \in \{601, 2003, 3901\}$, but the estimated probabilities were not markedly different for these other choices of $V$. With the exception of the IAC model, in this section we consider probabilities only for the complete ballot case.

\subsection{Descriptions of the Models}

We use two different types of models: one-dimensional spatial models and models which sample over the ballot simplex. In the spatial framework, we consider three different models. The first model, which we denote 1D Spatial, places three candidates and $V=1001$ voters on the real line using the standard normal distribution, and a voter's preference ranking is determined by their Euclidean distance from each candidate. For example, if the three candidates $A$, $B$, and $C$ are positioned by the normal distribution at $-1.2$, $-0.1$, and 1.1 respectively, then a voter positioned at 0.3 would cast a ballot with $B$ ranked first, $C$ ranked second, and $A$ ranked third. Such one-dimensional spatial models have been widely used in the social choice literature \cite{EH84, EH90, KGF, R23}.

Because of the increased political polarization in the United States and elsewhere, we also use two bimodal spatial models to investigate the potential effects of polarization on the probability of the reinforcement paradox. Under the second spatial model, which we denote 1D Bimodal, we create a bimodal distribution from an equally weighted combination of the normal distribution centered at $-2$ and the normal distribution centered at 2, both with standard deviation 1. See  Figure \ref{bimodal_figure} (top) for a visualization of this bimodal distribution. We then place candidates and voters along the line using the generated (normalized) bimodal distribution as a probability distribution, similar to 1D Spatial. Under the third spatial model, which we denote 1D Bimodal Weighted, we use the same two normal distributions to generate the bimodal distribution but give 50\% more weight to the left distribution, resulting in a bimodal distribution as shown Figure \ref{bimodal_figure}(bottom). We use this model to capture polarization when one side of the political spectrum is stronger than the other.
(We do not investigate the effect of changing the variance.)

\begin{figure}
\centering

\includegraphics[scale=0.55]{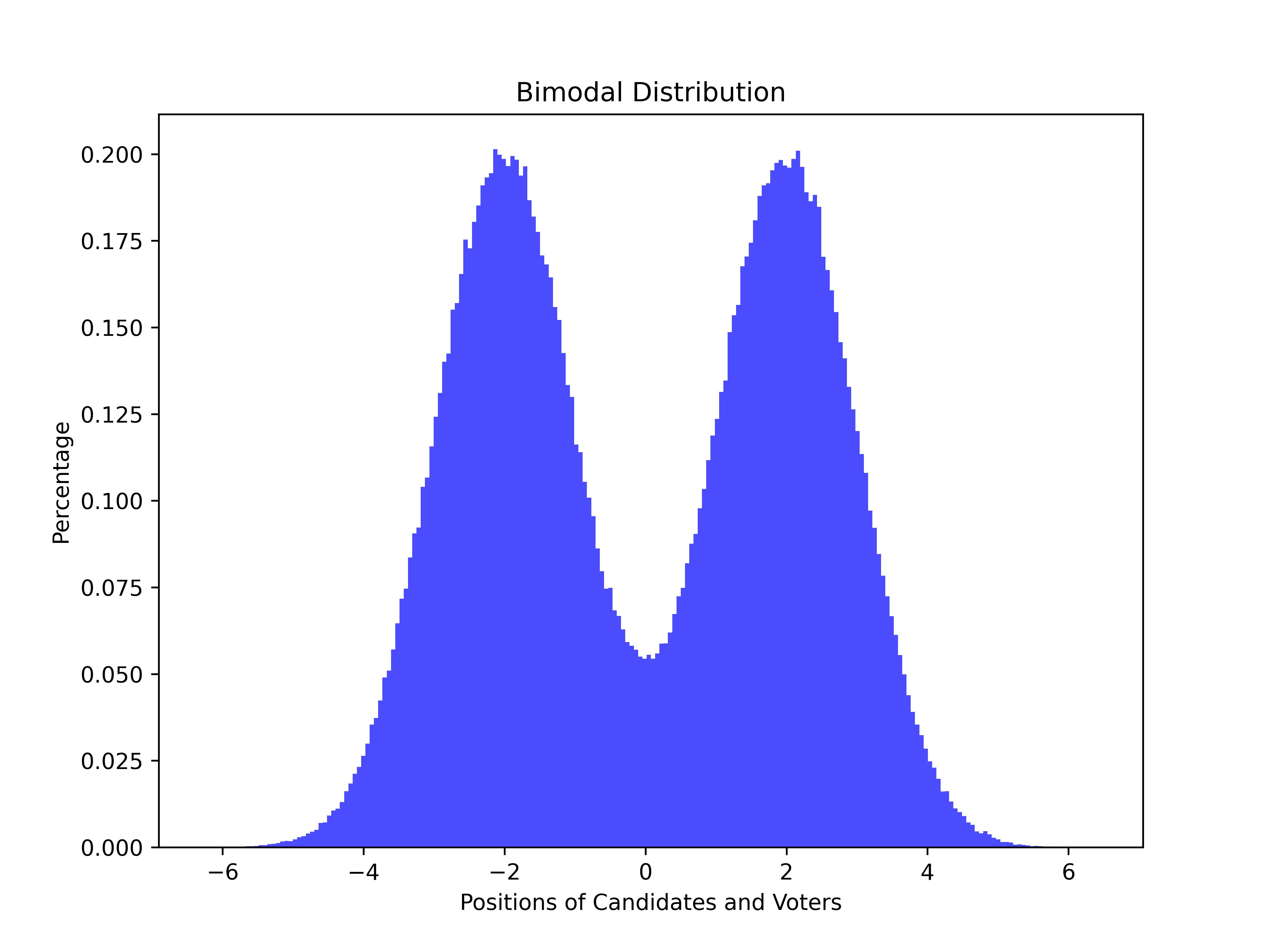}
\vspace{.05 in}

\includegraphics[scale=0.55]{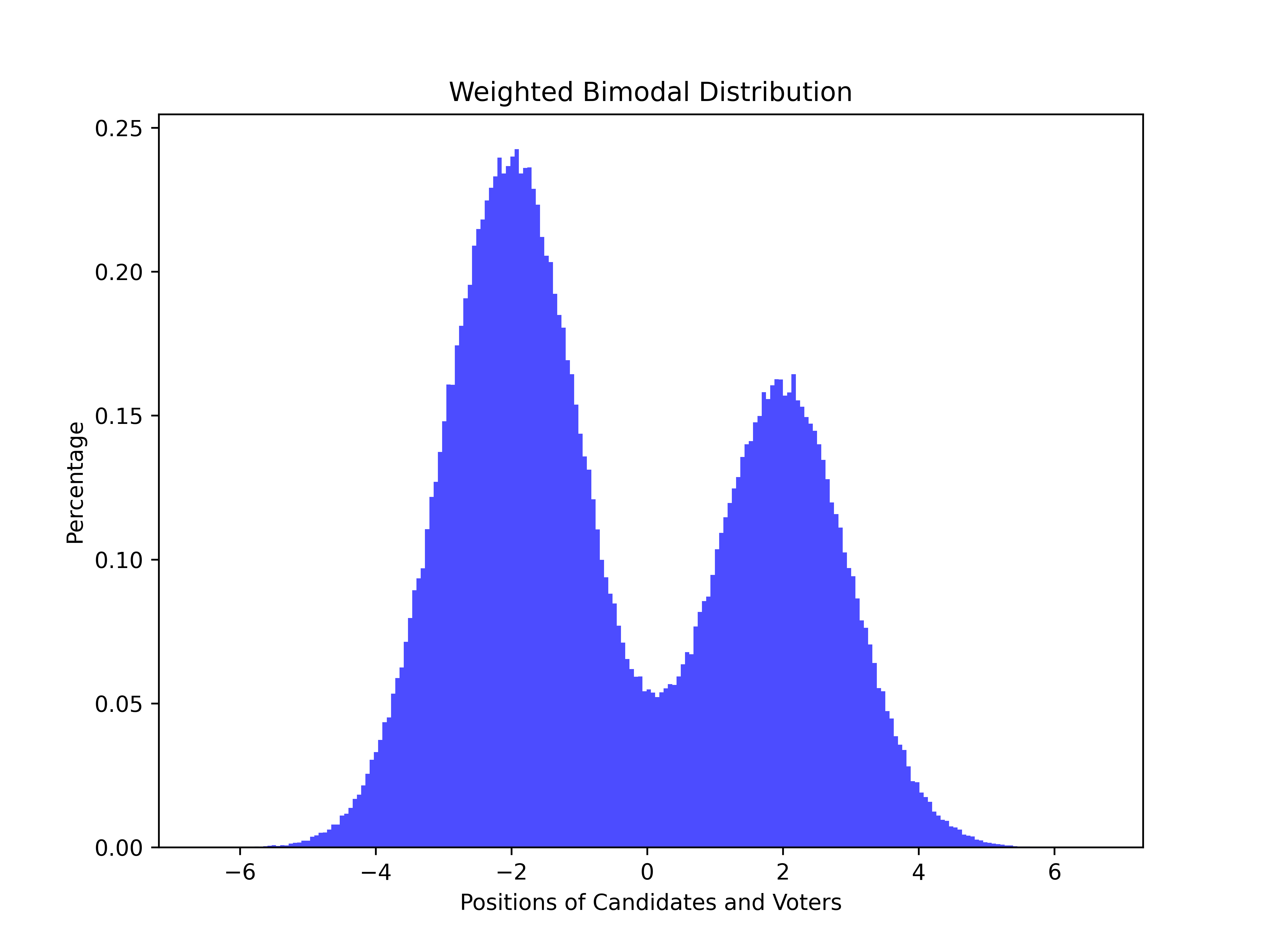}

\caption{Visualizations of the two bimodal spatial models.}
\label{bimodal_figure}
\end{figure}

For the ballot simplex models we use Dirichlet distributions with parameter $\alpha$, denoted Dir$(\alpha)$, and we use the implementation of these models provided by the software package VoteKit \cite{DDGGHMW}. (See also \cite{BGB21} for applications of the Dir$(\alpha)$ distribution to voting.) In a ballot simplex model, we place a probability distribution on the discrete ballot 5-simplex, which is the set of 6-tuples $(X_1, X_2, Y_1, Y_2, Z_1, Z_2)$ such that each number is non-negative and the six numbers sum to $V$. Under a Dirichlet model with $\alpha \in (0,\infty)$, as $\alpha \rightarrow 0$ the generated profile tends to weigh heavily in favor of a single candidate, so that the resulting election is a landslide victory for one of the candidates. For example, when $\alpha$ is close to zero, if $V=1001$ then a 6-tuple $(X_1, X_2, Y_1, Y_2, Z_1, Z_2)$ where $X_1+X_2>800$ is much more likely to be generated than a 6-tuple where each of the six numbers is approximately equal. When $\alpha = 1$ the model places a uniform distribution on the simplex, so that any 6-tuple of ranking numbers is as likely as any other. Dir$(1)$ is referred to as the \emph{impartial, anonymous culture} (IAC) model in social choice. As $\alpha$ increases the sampled 6-tuples tend to cluster in the center of the simplex, so that it becomes increasingly unlikely that a candidate earns a majority of first-place votes. In the limit $\alpha \rightarrow \infty$, the model is  the classical \emph{impartial culture} (IC) model. Under the IC model, each voter independently chooses one of the six preference rankings with uniform probability, leading to elections that, with high probability, are  closely tied. Thus, as $\alpha$ increases, a typical generated election becomes  highly competitive between the three candidates.

If we fix an $\alpha$,  we can consider the \emph{limiting probability} of observing a paradox where we let $V\rightarrow \infty$. In this limit, the discrete ballot simplex becomes the continuous 5-simplex determined by $X_1+X_2+Y_1+Y_2+Z_1+Z_2=1$, where each term in the sum is represents a (non-negative) proportion of votes. By using large electorates in our simulations, we estimate limiting probabilities under each model.

\subsection{Results}

Throughout this section and the next, in our analysis we identify four probabilities: the probability that a runner-up  (RU) paradox occurs; the probability that a plurality loser  (PL) paradox occurs; the probability that at least one of these  (RU or PL) paradoxes occurs; and the probability that both  (RU and PL) paradoxes occur.

We begin by providing  exact probabilities under Dir$(1)$, or IAC model. Because of its simplicity and prevalence in the social choice literature,  researchers have dedicated significant effort to developing tools for calculating probabilities under this model. One of the fruits of these labors is the software package Normaliz \cite{BISO},  which can provide limiting probabilities under IAC in three-candidate elections given a set of linear constraints.   In the complete ballot case, Normaliz determines the proportion of integer values of $X_1, X_2,  Y_1, Y_2,  Z_1$ and $Z_2$ that satisfy the inequalities in Proposition \ref{complete_cond} given a fixed number of voters $V$, as $V \rightarrow \infty$.  In the partial ballot case, Normaliz, does the same thing for $X_1, X_2, X_3, Y_1, Y_2, Y_3, Z_1, Z_2$ and $Z_3$ that satisfy the inequalities in Proposition \ref{partial_cond}. 

The probabilities in the following proposition are thus exact (rather than approximated through simulation) and are obtained directly from Normaliz by inputting the appropriate sets of inequalities. The conditional probabilities are obtained by adding an additional requirement to limit the number of elections to those without a majority candidate.

\begin{proposition}\label{IAC_prop}
The limiting probabilities that an election is susceptible to a reinforcement paradox under the Dir$(1)$ model (the IAC model) are as shown in Table \ref{IAC_results} for complete and partial ballots. The  ``Cond. Prob.'' probabilities are conditioned on the absence of a majority candidate.


\begin{table}[tbh]
\renewcommand{\arraystretch}{1.5}
\begin{tabular}{l|c|c|c|c}
&PL or RU & RU & PL & PL and RU\\
\hline
Prob. & $\frac{191}{768}\approx 24.9\%$  & $\frac{127}{576} \approx 22.5\%$ & $\frac{79}{1152} \approx 6.9\%$ & $\frac{31}{768} \approx 4.0\%$ \\
Cond. Prob.  & $\frac{191}{336}\approx 56.8\%$ & $\frac{127}{252} \approx 47.2\%$ & $\frac{79}{504}\approx 15.7\%$ & $\frac{31}{336}\approx 9.2\%$\\
\hline
\hline
Prob.& $\frac{1341257}{5308416} \approx 25.3\%$, & $\frac{527}{2304} \approx 22.9\%$ & $\frac{7031}{110592} \approx 6.4\%$ & $\frac{210439}{5308416} \approx 4.0\%$\\
Cond. Prob. & $\frac{1341257}{3006720}\approx 44.6\%$ & $\frac{527}{1305} \approx 40.4\%$ & $\frac{7031}{62640}\approx 11.2\%$ & $\frac{210439}{3006720}\approx 7.0\%$\\
\end{tabular}
\caption{Probabilities under the Dir$(1)$ (IAC) model. The top two rows correspond to the complete ballot case; the bottom two rows correspond to the partial ballot case. }
\label{IAC_results}
\end{table}
\end{proposition}

The  probabilities of one or more paradoxes occurring are similar in the complete and partial ballots cases. When conditioned on the absence of a majority candidate, the probabilities are significantly lower in the  partial ballot case. This is because  the probability an election does not contain a majority candidate is $145/256=56.64\%$ in the 8-simplex but is only $7/16=43.75\%$ in the 5-simplex. Overall, the probabilities are fairly high. This is largely because they  indicate the percentage of elections in which there exists a partition into two sub-elections demonstrating a reinforcement paradox; the numbers do not indicate anything about the number of such partitions in any given election.

For finite $\alpha \neq 1$  we are unaware of analytical techniques such as those that underlie Normaliz to calculate exact probabilities for distributions of the simplex other than IAC. Therefore, we use Monte Carlo simulations. Figure \ref{simulation_figure} illustrates the probabilities obtained from simulations for Dir$(\alpha)$ where $\alpha \in \{0.1, 0.2, 0.3, \dots, 4.0\}$. The top image  indicates unconditioned probabilities and the bottom image indicates probabilities conditioned on the absence of a majority candidate. The blue squares give the probability that an election is susceptible to a reinforcement paradox (of either type), the red disks give the probability an election is susceptible to a runner-up paradox, the black diamonds give the probability an election is susceptible to a plurality loser paradox, and the green pluses give the probability an election is susceptible to a double paradox.

\begin{figure}

\includegraphics[scale=0.8]{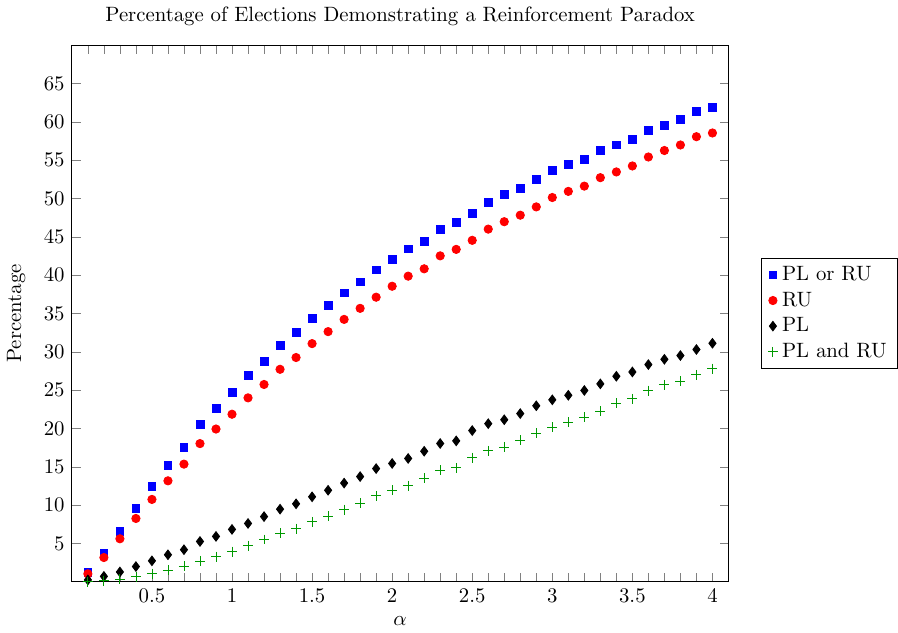}

\includegraphics[scale=0.8]{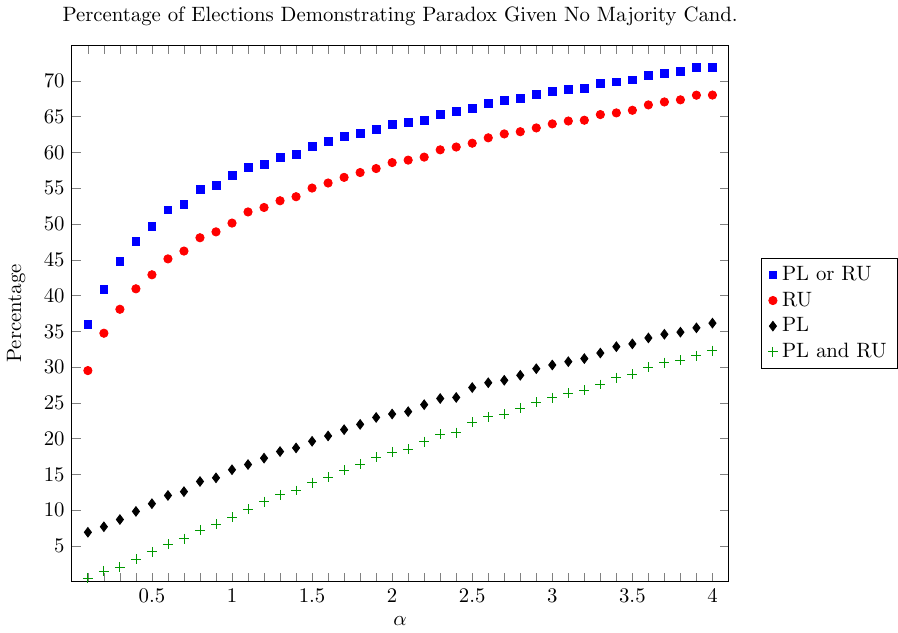}
\caption{Estimated probabilities that an election demonstrates a reinforcement paradox under the model Dir$(\alpha)$. The top figure shows unconditioned probabilities while the bottom figure shows probabilities conditioned on the absence of a majority candidate.}
\label{simulation_figure}
\end{figure}

As $\alpha$ increases, all four probabilities also increase; this is unsurprising because as $\alpha$ increases, a typical generated election becomes ``closer'' or ``more competitive.'' When $\alpha$ is small, many generated elections represent landslide victories where one candidate has a strong majority of first-place votes. When $\alpha$ is large, the generated elections tend to cluster  in the center of the simplex where there is no majority candidate. This also explains why we see a large increase in probabilities when moving from the top of Figure \ref{simulation_figure} to the bottom when $\alpha$ is small (when there are few elections without a majority candidate) but not when $\alpha$ is large (when almost all elections have no majority).

Table \ref{simluation_table} provides specific results from the simulations for a few choices of $\alpha$. Recall that 100,000 simulations were run for each choice of $\alpha$, so that when $\alpha=0.1$, for instance, only 3.5\% of generated elections contained no majority candidate. Note that we include simulations for $\alpha=1$ even though we have exact probabilities from Proposition \ref{IAC_prop}.

Table \ref{simluation_table} also displays information for three large $\alpha$ values which are not shown in Figure \ref{simulation_figure}. Recall that as $\alpha \rightarrow \infty$ the distribution resembles the IC model, in which each voter independently chooses one of the six possible rankings at random. For comparison, we also conducted simulations using the IC model  with $V=1001$, which returned an estimated probability of 1.0 for the runner-up, plurality, or both  (RU, PL, and RU and PL) paradoxes occurring. Intuitively, this is clear since  when every voter chooses a ranking at random, $X_1\approx X_2 \approx Y_1\approx Y_2 \approx Z_1\approx Z_2$. Thus, in the limit, Condition (2) and the plurality loser inequalities of Proposition \ref{complete_cond} will always be satisfied. It is unsurprising, therefore, that the probabilities in Figure \ref{simulation_figure} appear to be trending toward 1 and the values at the bottom of Table \ref{simluation_table}  indicate   convergence to 1 in all cases.

\begin{table}

\begin{tabular}{c|c|c|c|c|c}
$\alpha$ & PL or RU & RU & PL & PL and RU & No Maj. Cand.\\
\hline
0.1 & 1257 & 1033  &242  & 18 & 3500\\
0.5 & 12419 & 10737  &2730  & 1048 & 24024\\
1.0 & 24745 & 21852  & 6825 & 3932 & 43590\\
1.5 & 34319 & 31065  & 11082 & 7828 & 56455\\
2.0 & 42038 & 38548  & 15426 & 11936 & 65795\\
2.5 & 48079 & 44538  & 19721 & 16180 & 72663\\
3.0 & 53702 &50123   & 23732 & 20153 & 78324\\
\hline
10.0 & 85618& 83538 & 62924 &60844 &98498\\

20.0 &96383 &95392  & 86120 &85219 &99948\\

30.0&98814 & 98442  & 94367 & 93955 & 99997\\
\hline
$\infty$ (IC) & 100000 & 100000 & 100000 & 100000 & 100000\\

\end{tabular}
\caption{Simulation results for some choices of $\alpha$. For each $\alpha$ we generated 100,000 profiles.}
\label{simluation_table}
\end{table}

The 1D spatial models provide another way to estimate the probabilities of a reinforcement paradox, in part because the preferences they represent are single-peaked. Thus, in simulations in which the positions of the candidates $A$, $B$ and $C$ are ordered left to right, for instance,   $X_2=Z_1=0$ since $B$ must be either first or second-ranked by all voters. 

The results under all spatial models are shown in Table \ref{spatial_results}. The first row gives unconditioned probabilities; the second row gives probabilities conditioned on the absence of a majority candidate for 1D spatial. The middle (respectively last) two rows give the corresponding probabilities under 1D Bimodal (respectively 1D Bimodal Weighted).

Interestingly, the corresponding probabilities under 1D Spatial and 1D Bimodal are approximately equal, perhaps because both models are symmetric about zero. When we unbalance the distribution under 1D Bimodal Weighted, the probabilities change significantly. In particular, 1D Bimodal Weighted is the only model in this study under which it is more likely to observe a plurality loser paradox than a runner-up paradox. We do not know with certainty the reason for this, but we speculate that it is due to different responses to the so-called ``center squeeze'' \cite{P13}, where the centrist candidate is the Condorcet winner but is not the IRV winner. In one-dimensional three-candidate spatial models, the central candidate will often be the plurality loser even if they are closest to the median voter.  Using the notation of Section \ref{section:conditions}, if $C$ is the central candidate then $X_1=Y_1=0$ since $C$ must be second-ranked by all voters.  Hence the conditions in Case (1) of Proposition (1) cannot be satisfied (since they imply $B$ has the majority of the first-place vote). The conditions in Case (2) reduce to
$$Y>X-Z_2 \quad \mbox{ and } \quad Y>X_2+Z_1-Z_2= X+Z_1-Z_2.$$

If the distribution of voters is symmetric, $Z_1$ and $Z_2$ are likely to be similarly-sized, and so the rightmost inequality will be satisfied whenever $Y>X$.  Thus, a runner-up (RU) paradox will  occur often when the IRV and plurality winners are distinct.

If the distribution of voters is not symmetric, as with the  1D Bimodal Weighted  model,  the left and right-wing candidates are likely to be near the left and right peaks of the distribution respectively, with the centrist  in between but on the ``upward slope'' of the larger peak (the left one in Figure \ref{bimodal_figure}). Thus, the quantity $Z_1-Z_2$ is more likely to be larger than in the symmetric case, decreasing the likelihood of Case (2) occurring.  A similar argument holds if $A$ or $B$ is the central candidate.

For the plurality loser (PL) paradox to occur, it is easily seen that $C$ must be the central candidate. In this case, $C$ is again  likely to be on the upward slope of the larger peak in the 1D Bimodal Weighted model (and hence $Z$ is likely to be larger) than in the 1D Bimodal model. So the probability of the plurality loser (PL) paradox  occurring is greater  in the nonsymmetric case than the symmetric.

Under all three 1D spatial models, the probability of a double paradox (both RU and PL) is extremely small as it is rare for both $Y$ and $Z$ to be sufficiently large simultaneously.

\begin{table}
\begin{tabular}{l|c|c|c|c}
&PL or RU & RU & PL & PL and RU\\
\hline
Unconditioned Probability& 16.63\% & 12.17\% & 4.47\% & 0.01\%\\
Given No Majority Cand. & 40.07\% & 29.32\% & 10.76\% & 0.01\%\\
\hline
\hline
Unconditioned Probability& 16.68\% & 12.52\% & 4.18\% & 0.02\%\\
Given No Majority Cand. & 34.05\% & 25.55\% & 8.53\% & 0.03\%\\
\hline
\hline
Unconditioned Probability& 11.46 & 5.05 & 6.46 & 0.05\\
Given No Majority Cand. & 26.48 & 11.67 & 14.92 & 0.12\\
\end{tabular}
\caption{Estimated probabilities, reported as percentages, under our three spatial models. The top two rows give probabilities for 1D Spatial, the middle two rows for 1D Bimodal, and the last two for 1D Bimodal Weighted.}
\label{spatial_results}
\end{table}

\section{Empirical Results}\label{empirical}

In this section we provide empirical results using a large database of real-world IRV elections. Our data comes from three sources. First, we use single-winner political elections from the United States with at least three (not write-in) candidates. This data is available at the FairVote data repository \cite{O22}, and contains ballot data from municipal elections in cities such as San Francisco, CA and Minneapolis, MN. The repository also has ballot data for federal elections in the states of Alaska and Maine. Second, we use single-winner elections from the American Psychological Association (APA), which generally is willing to share ballot data for ranked-choice APA presidential elections and elections for the Board of Directors. Some of this data is available on preflib.org \cite{Preflib}, and some was shared directly with the first author. Third, we use single-winner political elections from a database of Scottish local government elections. The ballot data for many of these elections is available at \url{https://github.com/mggg/scot-elex}. The Scottish elections at this repository occurred in 2022 or before. We have separately collected preference profiles for several off-cycle IRV by-elections which have occurred since 2022. These profiles are available by request.
 
 In total, we accessed the preference profiles for 416 single-winner IRV elections with at least three candidates.\footnote{Any reference to the number of candidates in an election does not include write-ins. In our dataset, write-in candidates receive an insignificant amount of votes and ignoring them does not affect the analysis.} Many of these elections contain more than three candidates; because we are interested in the three-candidate case, for each election we ``reduce'' the number of candidates by running the IRV algorithm until only three candidates remain and then analyze the resulting profile. Note that  the relationship between the original election and the resulting three-candidate election is not clear with respect to the susceptibility to  reinforcement paradoxes. It is \emph{a priori} possible that the original election demonstrates a paradox while the reduced three-candidate election does not, and vice versa. We use this methodology because we are interested in three-candidate elections and analyzing the elections with more than three candidates using  the IRV algorithm in this  way allows us to expand the size of our three-candidate election dataset. Moreover, in some cases if the resulting three-candidate is susceptible to a paradox then we can gain insight into how the original election might be susceptible to a paradox; this is how a reinforcement paradox in a five-candidate Minneapolis city council election was found, for example \cite{MM23}.

 \begin{table}[tbh]
\begin{tabular}{p{6.2cm}|cccc}
Jurisdiction & USA & APA & Scotland & Total \\
\hline
Elections & 331 & 27 & 58 & 416 \\
\hline
3-Cand. Elections & 121 & 10 & 1& 132 \\
\hline
3-Cand. Elections, no Maj. Cand. & 57 & 6 & 1 & 64 \\
\hline
3-Cand. Reduced Elections, no Maj. Cand.  & 206 & 20 & 48 & 274 \\
\end{tabular}
\caption{A summary of our data sources.}
\label{data_info}
\end{table}

\subsection{Results}

Table \ref{empirical_results} provides the results for elections having three candidates in the original election and for all  416 elections. We use the same abbreviations RU and PL as in Section \ref{models_section}. The first row indicates the number and percentage of elections exhibiting one or more paradoxes among the three-candidate elections. For example, of the 132 elections which originally contained three candidates,  19 of them (14.4\%)  exhibit a reinforcement paradox. The second row indicates the same information but for all the elections including  the reduced elections.   The third and fourth rows indicate the percentage of elections that exhibit the given type of paradox, conditioned on the absence of a majority candidate. 

Unlike the simulated elections in Section \ref{models_section}, all the elections in the database  that demonstrate a plurality loser paradox also demonstrates a runner-up paradox. (Note that the  two rightmost columns of Table \ref{empirical_results} are identical.) The most likely explanation is that in real elections  runner-up candidates tend to be  much stronger than plurality losers.

 \begin{table}
\begin{tabular}{p{4.8cm}|c|c|c|c}
 & PL or RU & RU & PL & PL and RU\\
\hline

 3-Cand. Elections  & 19, 14.4\% &   19, 14.4\% & 7, 5.3\% &  7, 5.3\%\\

 3-Cand. Reduced Elections & 103, 24.8\% & 103, 24.8\% & 30, 7.2\%  & 30, 7.2\% \\

 \hline
 \hline

 3-Cand. Elections, no Maj. Cand. & 29.7\% & 29.7\% & 10.9\% & 10.9\% \\

3-Cand Reduced Elections, no Maj. Cand. & 37.6\% & 37.6\%  & 10.9\%  & 10.9\%\\
\end{tabular}
\caption{Reinforcement paradox results for the 132 real-world  three-candidate elections and for all 416  ``reduced'' elections.}
\label{empirical_results}
\end{table}

In Section \ref{threshold_results}, we showed in Proposition \ref{threshold_results} that if an election demonstrates a runner-up reinforcement paradox then the runner-up could receive as little as 25\% of the first-place votes and the plurality loser could receive an arbitrarily small percentage of the first-place votes. Similarly, if an election demonstrates a plurality loser paradox then the plurality loser can receive as little as 25\% of the first-place votes. How close to these minima do we come in real-world elections? The lowest percentage of first-place support earned by the plurality loser in an election which demonstrated a reinforcement paradox is 3.6\%, which occurred in the 2023 city council ward 7 election in Minneapolis. Table \ref{Minn_Ward7_profile} shows the preference profile for this election, which  originally contained three (not write-in) candidates; that is, the profile displayed is not obtained from eliminating previous official candidates. The three candidates were Cashman ($C$), Graham ($G$), and Foxworth ($F$). Foxworth earns only 280 (3.6\%) first-place votes of the 7969 ballots cast, yet the runner-up inequalities in Condition (2) of Proposition \ref{partial_cond} are satisfied.

\begin{table}

\begin{tabular}{l|ccccccccc}
Num. Voters & 1440 & 1217 & 1212 & 1473 & 611 & 1726 & 185 & 68 & 37\\
\hline
1st & $C$ &$C$ &$C$ & $G$ &$G$ &$G$ &$F$ &$F$ &$F$\\ 
2nd& $G$ & $F$ & & $C$ & $F$& & $C$ & $G$ & \\
3rd & $F$ & $G$ & & $F$ & $G$ & & $G$ & $C$ & \\

\end{tabular}
\caption{The preference profile for the 2023 city council ward 7 election in Minneapolis.}
 \label{Minn_Ward7_profile}
\end{table}

This 2023 Minneapolis city council election is an interesting example of a reinforcement paradox, even though it initially appears unremarkable if we are hunting for voting paradoxes. In this election the plurality, Condorcet, and IRV winners are identical, ruling out the possibility of classic issues like a monotonicity paradox, spoiler effect, or similar pathology. While the election is close—Cashman and Graham receive a comparable number of first-place votes—Cashman is the clear winner. From a social choice perspective, this election typically warrants no further analysis. However, the presence of a reinforcement paradox makes this election more compelling than it might appear at first glance.

The lowest first-place vote percentage earned by the IRV runner-up in an election which demonstrates a runner-up paradox is 28.9\%, achieved in the 2018 mayoral election in San Francisco. This election had 9 official candidates, and thus the 28.9\% represents the vote share of the runner-up after eliminating six candidates. The lowest percentage for a plurality loser in an election which demonstrates a plurality loser election is 26.9\%, achieved in the 2008 County Executive election in Pierce County, Washington. This election had four official candidates, and the plurality loser earned 26.9\% of the voter after eliminating one candidate. 

\subsection{Precinct and District Partitions}\label{section:precincts}

As discussed in Section \ref{prelims}, some instances of a reinforcement paradox are more serious than others. Our approach thus far has been to allow for arbitrary partitions of the electorate, including partitions in which the two subsets of voters share no meaningful characteristics. We take this approach because in many cases the only available data is a preference profile, which provides no information beyond aggregate preference rankings. Ideally, we would search for partitions defined by meaningful variables; for example, we could attempt to find instances where candidate $A$ wins among voters with a college degree and among voters without a college degree, but loses overall. However, our data does not contain such information. Some of our data does include limited geographic information, such as the precinct at which a voter cast their ballot. Some data files also include district information, where a ``district'' is a collection of precincts. For example, election files from Alaska indicate a voter's precinct and their state House district. In this section, we restrict attention to partitions of the electorate that respect precinct or district boundaries.

In the 416 real-world elections comprising our dataset, none of the APA or Scottish elections contain any information beyond the preference profile. Of the American political elections, 304 contain precinct information and 84 indicate a voter's precinct and district.  Thus, for 304 elections we can search for precinct partitions and for 84 we can search for district partitions.

If an election contained 16 or fewer precincts then we exhaustively searched through all precinct partitions and checked for a paradox. Of course, this brute force approach is not feasible for a large number of precincts, and some of our elections contained hundreds of precincts. Thus, we instead formulate the search for reinforcement paradoxes as a binary integer linear feasibility problem.

Our methodology is as follows. First, we reduce the cast vote record\footnote{A \emph{cast vote record}, in contrast to a preference profile, is a record of each voter's preference ranking of the candidates. We use this object for our computations in this section because the preference profile might combine voters across precincts, potentially losing geographic information.} to three candidates, as we did in previous sections.  We then construct a precinct-level summary of the reduced election. For each precinct, we record the number of first-place votes received by each of the three candidates, as well as the number of ballots that would transfer from one candidate to another upon elimination. These quantities are sufficient to determine the outcome of any two-round elimination process among the three candidates within any subset of precincts.
Next, we reduce the search for an aggregation paradox to a feasibility problem. Specifically, we construct a binary integer linear program (ILP) in which each decision variable indicates whether a given precinct is assigned to one of two subsets. Linear constraints are imposed to ensure that both subsets are nonempty and that, within each subset, the induced vote totals satisfy the inequalities corresponding to a prescribed IRV elimination order. These constraints encode both the identification of the first-round loser and the outcome of the final round after vote transfers.

A feasible solution to this ILP therefore corresponds to a partition of precincts into two subsets such that each subset yields a specified elimination order, potentially differing from the outcome of the election as a whole. We solve this feasibility problem using a standard branch-and-bound algorithm. If a feasible solution exists, we conclude that the election exhibits an aggregation paradox of the specified type. We independently verified the precinct partition results produced by this methodology that are presented in the article. See \cite{BW05, W98} for more details about the implementation of ILPs.

Table \ref{table:precinct_district_results} provides our precinct and district results. The percentages of paradoxes are much smaller than when we look for arbitrary partitions, suggesting that the more serious kind of reinforcement paradox is much less common than instances in which we can choose an arbitrary partition. For several of the elections demonstrating a precinct or district-level paradox, we were unable to find a ``nice'' partition of precincts or districts corresponding to a north-south kind of outcome. However, we identified several elections with relatively clean geographic partitions which we discuss below.

To streamline notation, we write $\mathcal{B}_1$ as a collection of precincts and let $\mathcal{B}_2$ be the precincts not in $\mathcal{B}_1$, which we will denote as a complement $\mathcal{B}_2=\mathcal{B}_1^C$.
Two of the four examples below come from reduced elections which originally had more than three candidates. In these cases, we also determine whether  the partition produces a paradox with all candidates still present.  We indicate if keeping all candidates does not produce a paradox. 

In the first three examples,  the number of precincts was small enough to manually check all possible precinct partitions. In each case, there were multiple precinct partitions that resulted in a runner up paradox but none that resulted in a plurality loser paradox. 

\begin{table}
\centering
\begin{tabular}{l|c|c|c|c}
&PL or RU & RU & PL & PL and RU\\
\hline
Precinct & 9, 3.0\% & 8, 2.6\% & 1, 0.3\% & 0\\
\hline
District & 1, 1.2\% & 0 & 1, 1.2\% & 0 \\
\end{tabular}
\caption{The number of elections demonstrating a reinforcement paradox when using a partition of precincts or districts. The corresponding percentage uses a denominator of 304 for the first row and 84 for the second.}
\label{table:precinct_district_results}
\end{table}

\begin{table}
\centering
\begin{tabular}{ll|ccccccccc}
\multicolumn{2}{c|}{Num. Voters} 
& 723 & 1260 & 636 
& 1076 & 141 & 1191 
& 1475 & 126 & 1778 \\
\hline
\multicolumn{2}{c|}{1st choice}  
& $C$ & $C$ & $C$ 
& $D$ & $D$ & $D$ 
& $SY$ & $SY$ & $SY$ \\
\multicolumn{2}{c|}{2nd choice}  
& $D$ & $SY$ &  
& $C$ & $SY$ &  
& $C$ & $D$ &  \\
\multicolumn{2}{c|}{3rd choice}  
& $SY$ & $D$ &  
& $SY$ & $C$ &  
& $D$ & $C$ &  \\
\hline
\hline
\multirow{2}{*}{Partition} 
& $\mathcal{B}_1$ 
& 372 & 589 & 340 
& 573 & 65 & 723 
& 633 & 50 & 641 \\
& $\mathcal{B}_2$ 
& 351 & 671 & 296 
& 503 & 76 & 468 
& 842 & 76 & 1137 \\
\end{tabular}
\caption{Three-candidate reduced preference profile for the 2021 Portland, ME At-Large Seat 1 Commissioner election. $\mathcal{B}_1$ represents voters from precincts $\{$Absentee, 4-2, 5-1, 5-2$\}$ and $\mathcal{B}_2$ represents voters from all other precincts. Sheikh-Yousef ($SY$) wins using ballots from $\mathcal{B}_1$ or from $\mathcal{B}_2$ but loses the overall election to Chann ($C$).}
\label{table:Portland_ME_table}
\end{table}

\begin{example} \textbf{2021 Portland, ME At-Large Seat 1 Commissioner Election.} This election is an outlier in our dataset as it was a multi-winner election using a voting method  called ``sequential ranked choice voting'' or ``sequential IRV'' to elect four winners from an initial pool of eleven candidates. Under this method, the first seat is allocated using IRV, and the winner is then removed from the preference profile. The second seat is then allocated to the IRV winner using the modified profile, after which this second winner is also removed. The process continues until all seats are filled. (See \cite{MMLS24} for an in-depth analysis of this voting method.) 

Although the election involved multiple winners, we include the election for the first seat  in our dataset because the winner was chosen using IRV, and there is arguably some psychological or other non-tangible benefit to winning the first seat. In the election for the first seat, Marpheen Chann was the IRV winner and Nasreen Sheikh-Yousef was the runner-up.

Portland has eleven precincts (see Figure \ref{figure:portland_ME} in the Appendix). In addition, the City of Portland also groups all absentee ballots into their own precinct with no obvious geographic identification.  Aside from this ambiguity, the two most geographically interesting partitions are

\begin{itemize}
\item $\mathcal{B}_1=\{$Absentee, 4-2, 5-1, 5-2$\}$, $\mathcal{B}_2=\mathcal{B}_1^C$. This partition is described in Table \ref{table:Portland_ME_table}.
\item $\mathcal{B}_1=\{$Absentee, 1-1, 1-2, 1-3, 2-1$\}$, $\mathcal{B}_2=\mathcal{B}_1^C.$
\end{itemize}

If we count Absentee ballots as coming from the northwest then in the first partition, Sheikh-Yousef won the northwest part of Portland and she won the non-northwest part, yet she lost the first seat. Similarly, if we count Absentee ballots as coming from the east then in the second partition, Sheikh-Yousef won the east and she won the west, yet again she loses overall.
    
\end{example}

\begin{example}\textbf{2022 Oakland District 4 School Direction Election.} This election contained the candidates Mike Hutchison (IRV winner), Nick Resnick (IRV runner-up), and Pecolia Manigo (plurality loser). Figure \ref{figure:Oakland_D4} in the Appendix shows the precinct map of District 4, which contains sixteen precincts. The four most geographically interesting partitions generating a runner up paradox are 
\begin{itemize}
\item $\mathcal{B}_1=\{$213100, 213500$\}$, $\mathcal{B}_2=\mathcal{B}_1^C$ 
\item$\mathcal{B}_1=\{$213500$\}, \mathcal{B}_2=\mathcal{B}_1^C$
\item$\mathcal{B}_1=\{$245300$\},\mathcal{B}_2=\mathcal{B}_1^C$
\item$\mathcal{B}_1=\{$345900$\},\mathcal{B}_2=\mathcal{B}_1^C.$
\end{itemize}

In the first partition, $\mathcal{B}_1$ consists of the two northernmost precincts. Thus, this partition represents a runner-up paradox with a clear north-south split in which Resnick wins the northern voters and separately wins the non-northern voters, yet loses the overall election. Similarly, the other three partitions represent instances when Resnick wins in an individual precinct, as well as winning outside the precinct, yet loses the overall election.

\end{example}

The next example was analyzed in \cite{MM23}. Reinforcement paradoxes are not the focus of that article, but the authors point out a reinforcement paradox can be achieved using the precinct partition $\mathcal{B}_1=\{$2-1, 2-5$\}$, $\mathcal{B}_2=\mathcal{B}_1^C$. The two partitions given below are more geographically interesting.

\begin{example} \textbf{2021 Minneapolis Ward 2 City Council Election.} This election contained the candidates Tom Anderson, Yusra Arab, Cam Gordon, Guy Gaskin, and Robin Worlobah. Anderson and Gaskin received relatively few first-place votes and were eliminated in the first two rounds. With three candidates remaining, the IRV winner is Worlobah, the runner-up is Arab, and the plurality loser is Gordon.

Figure \ref{figure:Minneapolis_Ward2} in the Appendix shows the precinct map of Ward 2 which contains eleven precincts. The two most geographically interesting partitions generating a runner up paradox are 
\begin{itemize}
\item $\mathcal{B}_1=\{$2-3, 2-4, 2-8, 2-11$\},\mathcal{B}_2=\mathcal{B}_1^C$ 
\item$\mathcal{B}_1=\{$2-1, 2-2, 2-3, 2-5, 2-6$\}, \mathcal{B}_2=\mathcal{B}_1^C.$
\end{itemize}

The first partition produces a runner-up paradox in the reduced three-candidate election and in the original election with all five candidates.  In that partition, $\mathcal{B}_1$ represents a coalition of northernmost and southernmost precincts with $\mathcal{B}_2$ representing the second reamining voters.  Thus, in this election we can say that Arab wins the northernmost and southernmost voters, and wins separately with all other voters, but she loses the actual election.
The second partition produces a paradox only in the reduced three-candidate election.  

\end{example}

\begin{table}
\centering
\resizebox{\textwidth}{!}{
\begin{tabular}{ll|ccccccccc}
\multicolumn{2}{c|}{Num. Voters} 
& 27070 & 15478 & 11262 
& 34078 & 3659 & 21237 
& 47419 & 4647 & 23733 \\
\hline
\multicolumn{2}{c|}{1st choice}  
& $B$ & $B$ & $B$ 
& $Pa$ & $Pa$ & $Pa$ 
& $Pe$ & $Pe$ & $Pe$ \\
\multicolumn{2}{c|}{2nd choice}  
& $Pa$ & $Pe$ &  
& $B$ & $Pe$ &  
& $B$ & $Pa$ &  \\
\multicolumn{2}{c|}{3rd choice}  
& $Pe$ & $Pa$ &  
& $Pe$ & $B$ &  
& $Pa$ & $B$ &  \\
\hline
\hline
\multirow{2}{*}{Partition} 
& $\mathcal{B}_1$ 
& 6148 & 3328 & 2977 
& 9506 & 1116 & 7025 
& 7895 & 830 & 3712 \\
& $\mathcal{B}_2$ 
& 20922 & 12150 & 8285 
& 24572 & 2543 & 14212 
& 39524 & 3817 & 20021 \\
\end{tabular}}

\caption{An example of a state House district partition demonstrating a plurality loser paradox in the 2022 Alaska Special House election.}
\label{table:AK_partition}
\end{table}
We found one election demonstrating a plurality loser (but not runner-up) paradox at the precinct and district level: the 2022 Special House election in Alaska. We analyze this election in our final example.

\begin{example} \textbf{2022 Alaska Special House Election.} This election contained the three candidates  Mary Peltola (IRV winner), Sarah Palin (IRV runner-up), and Nick Begich (plurality loser). Since the voter preference file contains state House district data as well as precinct information, we searched for partitions involving both of these geographic units. 

Alaska has 40 state House districts as well as a 41st ``district'' consisting of overseas ballots. There are hundreds of partitions of the House districts demonstrating a paradox where Begich wins each sub-election, but none where Palin does. (The same result holds for precinct partitions.) A likely explanation is that Begich was the election's Condorcet winner and thus the election's ``strongest candidate'' in some sense, while Palin was the Condorcet loser.

 We could not find a clean north-south type of partition of the House districts demonstrating the plurality loser paradox. An example of an aesthetically pleasing partition can be found in the top of Figure \ref{figure:AK_map} in the Appendix, corresponding to the partition detailed in Table \ref{table:AK_partition}. The black region corresponds to the district partition $\mathcal{B}_1= \{$Overseas,  1, 7, 8, 18, 22, 26, 30, 32, 33, 37, 38,  40$\}$. The gray region  corresponds to  $\mathcal{B}_2=\mathcal{B}_1^C$. We could not find a partition in which the two sub-electorates  both represented contiguous regions, mainly because House district 1 in the southeast corner of Alaska tended to be paired with regions very far away, such as districts around Fairbanks or the far-north Arctic district, creating a gap between District 1 and the rest of the districts in its partition subset. The partition we found which comes closest to producing two contiguous regions is shown in the bottom of Figure \ref{figure:AK_map}, where Begich wins the black districts and also wins the gray districts grouped with the overseas ballots.
\end{example}

 This Alaska election is one of only two single-seat elections in our American dataset in which the Condorcet winner is not the IRV winner, perhaps explaining why there are geography-based plurality loser paradoxes but none for the runner-up (who is the Condorcet loser).  The other American election in our dataset in which the Condorcet winner is the plurality loser (the 2009 mayoral election in Burlington, VT) does not have precinct or district information. This election does demonstrate runner-up and plurality loser paradoxes but we cannot check for ``natural'' or partitions of the electorate.

\subsection{Bootstrapping Results}

In Section \ref{models_section} we ran simulations based on sampling using theoretical models such as the IAC model. Since we have access to a large dataset of real-world elections, we can also sample from these elections and run simulations using nonparametric bootstrap analysis.  Such sampling from real-world elections allows for an empirically-grounded version of the sampling methodology in \ref{models_section}.

Bootstrapping allows us to generate random samples from each election profile  to create a set of ``pseudoprofiles.'' Analysis of these pseudoprofiles can give a sense of how robust the presence or absence of a paradox is under perturbations of the election profile. For instance, if we sample from an election profile that does not demonstrate a reinforcement paradox and determine that none of the sampled pseudoprofiles demonstrates a paradox, then the original election is not close to an election that demonstrates a paradox. On the other hand, if an election demonstrates a paradox but relatively few of its pseudoprofiles do, this indicates that the paradox is sensitive to small perturbations and is unlikely to persist under resampling. This is the case, for example, in the 2023 city council ward 7 election in Minneapolis (Table \ref{Minn_Ward7_profile}). The election demonstrates a reinforcement paradox but only 337 of 1000 generated pseudoprofiles do. Thus, very similar elections are unlikely to demonstrate a paradox.

To perform this analysis, we follow techniques used in the classical papers \cite{PPR} and \cite{RKKH}. For each of the 416 three-candidate profiles  from our real-world dataset, we created 1,000 pseudoprofiles by sampling ballots with replacement from the actual three-candidate profile. This sampling allowed us to generate 416,000 three-candidate elections which we can check for  reinforcement paradoxes.  

To keep the computation time reasonable and to match  the analysis in Section \ref{models_section}, if the actual profile contained more than 1,001 voters then for each generated pseudoprofile we sampled only 1,001 ballots; otherwise each pseudoprofile contained the same number of voters as the original election. Most of the elections in the dataset contain more than 1,001 voters and thus most of the generated elections contain 1,001 voters.

Table \ref{bootstrapping_results} shows the probability of each type of reinforcement paradox across the 416,000 pseudoprofiles. The results are similar to the ``3-Cand Reduced Elections'' probabilities in Table \ref{empirical_results}, and thus  the bootstrapping results do not meaningfully affect how frequently we expect such paradoxes to occur in real-world elections.

A more nuanced summary of the bootstrapping results can be seen   in  Figure \ref{bootstrap_histograms}. Each histogram illustrates the frequency of the proportion of pseudoprofiles generated by each original election that resulted in a ``paradoxical'' election. The top image is based on the 103 original elections that demonstrate the paradox. Thus, for example, the right-most bar on the top histogram indicates that in 41  of these elections, between 95\% and 100\% of the associated pseudoprofiles were susceptible to the paradox. Similarly, the second-to-the-right bar  indicates that in 13 of these elections, between 90 and 95\% of the associated pseudoprofiles were susceptible to the paradox.  The bottom histogram is constructed the same way for the 313 original three-candidate elections which were not susceptible to  the paradox with the exception that elections that generated zero ``paradoxical'' pseudoprofiles  are omitted for clarity (since the left bar would be out-sized).

From Figure \ref{bootstrap_histograms} (top), we can observe that roughly half of the original elections that were susceptible to a paradox were extremely likely to exhibit that paradox in the sense that more than 90\% of their associated pseudoprofiles also did so. Most of the remaining elections in this category also generated pseudoprofiles that  more often than not were susceptible to a reinforcement paradox. The reverse is true in Figure \ref{bootstrap_histograms} (bottom): the large majority of the original elections that were not susceptible to  the paradox generated few pseudoprofiles that were  susceptible to a reinforcement paradox.

Noteworthy, however, is the fact that across all 416 original elections, 241 generated at least one pseudoprofile which was susceptible to a reinforcement paradox. While the presence of a single such pseudoprofile should be interpreted with caution—since rare events may arise in large bootstrap samples—this nevertheless indicates that paradoxical outcomes can occur under perturbations of many of the observed elections. If we restrict the analysis to elections that did not contain a majority candidate, 224 of the 254 original elections generated at least one pseudoprofile that is susceptible to a reinforcement paradox. These results suggest that, although such outcomes may occur with low frequency in some cases, many three-candidate real-world elections without a majority candidate admit generated pseudoprofiles under which a reinforcement paradox can arise. We note that such a high percentage is not achieved under bootstrapping for other paradoxes. For example, in \cite{MW23} the same type of analysis was performed for the spoiler effect using a smaller dataset of American IRV elections (with 10,000 pseudoprofiles per election), and only five elections generated at least one pseudoprofile demonstrating a spoiler effect under IRV. Thus, it is not the case that performing bootstrap analysis guarantees that rare paradoxical outcomes will appear broadly across elections.

Finally, recall that in the 416 actual elections, any election that was susceptible to a plurality loser paradox was also susceptible to a runner-up paradox. This is not the case for pseudoprofiles: 1,277 generated elections are susceptible to a plurality loser but not to a runner-up paradox. Thus, while such outcomes are relatively rare, the bootstrapping analysis shows that they can arise under perturbations of real-world elections, even though they are not observed in the original data.

\begin{table}
\begin{tabular}{l|c|c|c|c}
&PL or RU & RU & PL & PL and RU\\
\hline
Unconditioned Prob.& 25.0\% &24.8\% & 6.8\% & 6.5\%\\
Conditioned Prob. & 42.3\% & 41.8\% &11.5\%   & 11.0\%\\

\end{tabular}
\caption{Probability, expressed as a percentage, of observing a reinforcement paradox in the 416,000 generated pseudoprofiles. The second row gives probabilities conditioned on the absence of a majority candidate.}
\label{bootstrapping_results}
\end{table}

\begin{figure}

\centering

\includegraphics[scale=0.5]{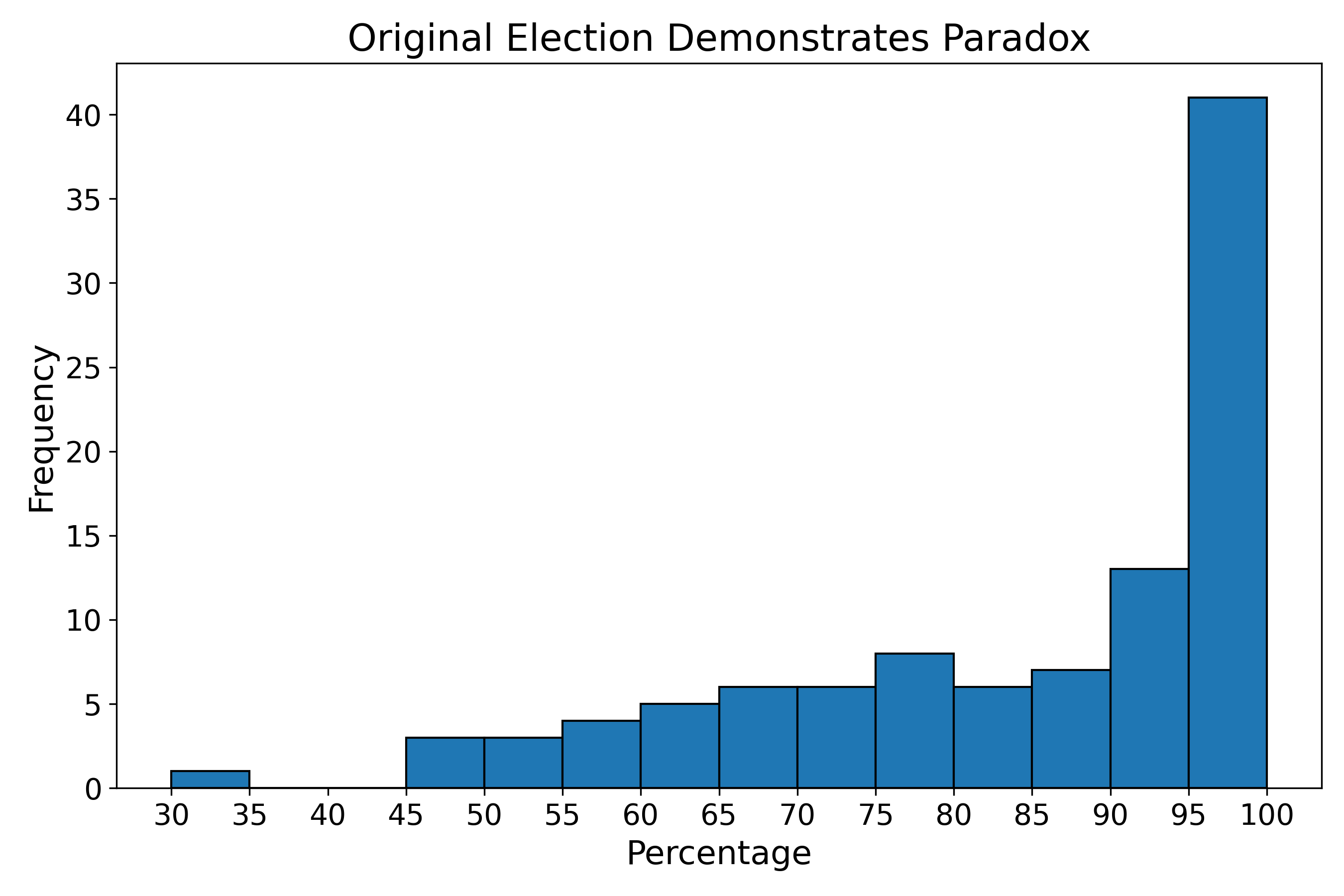}
\vspace{.1 in}

\includegraphics[scale=0.5]{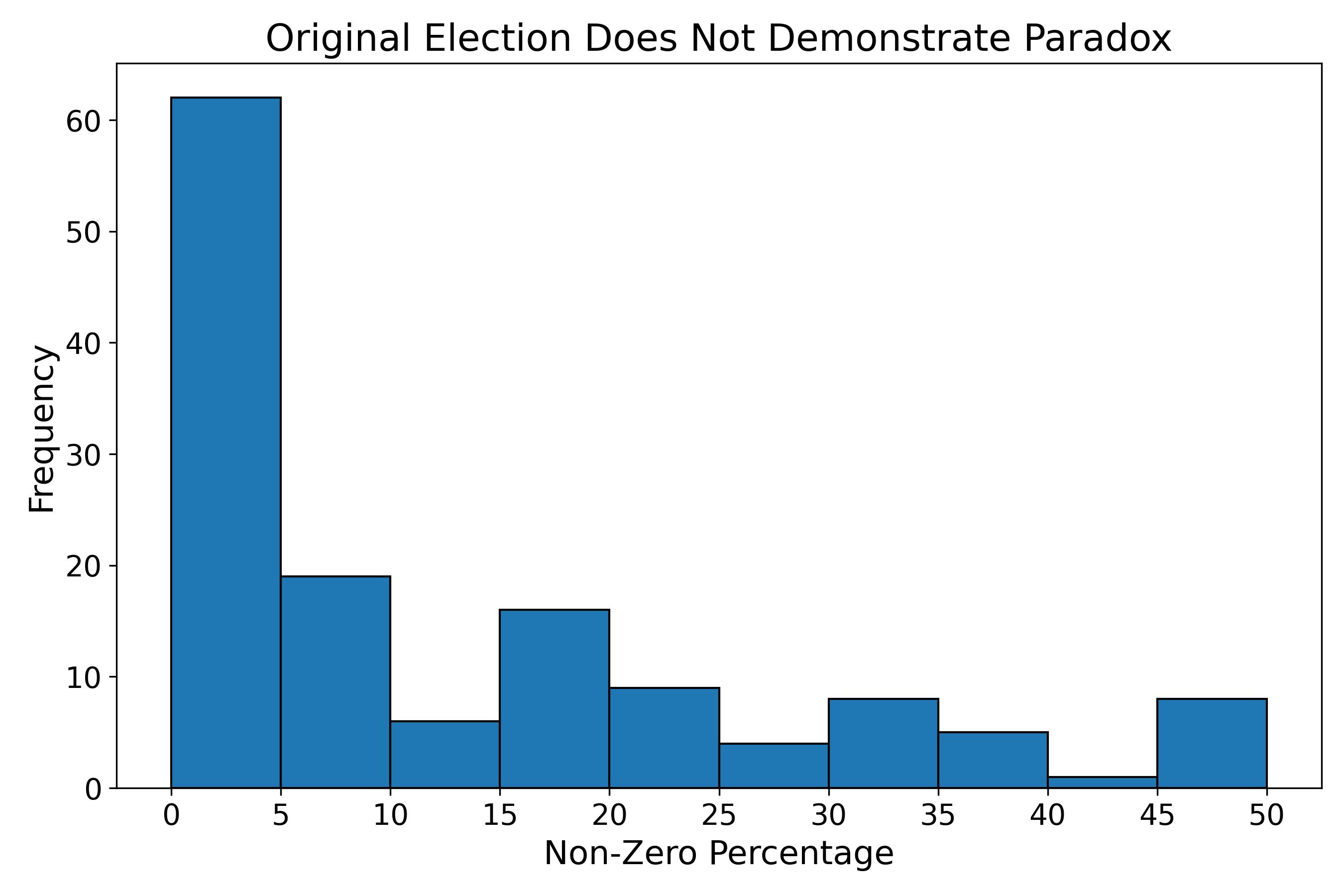}

\caption{(Top) Bootstrapping frequency results for elections that are susceptible to a reinforcement paradox. (Bottom) Bootstrapping frequency results for elections that are are not susceptible to a reinforcement paradox but generated at least one pseudoprofile that is.}
\label{bootstrap_histograms}
\end{figure}

\section{Discussion}\label{section:discussion}

In this section we discuss our results, starting with a comparison between the theoretical and empirical results, then placing our results into context by comparing them  to previous results on the reinforcement paradox. 

\subsection{Theoretical Versus Empirical Results} As is often the case in social choice, none of the theoretical models give probabilities which closely match what we observe in real-world elections across all four paradox types. Interestingly, if we focus on the 416 three-candidate reduced elections (lines two and four of Table \ref{empirical_results}), the IAC complete-ballot model (Table \ref{IAC_results}) provides a decent approximation of the empirical results.  Furthermore, even though the real-world elections contain many partial ballots, generally the complete-ballot IAC model more accurately reflects the empirical probabilities across the 416 reduced elections than the partial-ballot IAC model. For other voting paradoxes such as monotonicity paradoxes or the spoiler effect, the IAC model typically gives probabilities much higher than what we observe in real elections. If we focus only on the 132 actual three-candidate elections, a model like Dir$(0.6)$ gives the best approximation of the empirical unconditioned results (with the exception of ``PL and RU''), which makes sense because in these elections, the winner is stronger overall than in the Dir$(1)$ (IAC) model, consistent with the election data. Overall, the 1D-spatial models, both normal and bimodal,  underestimate the probabilities of one or more types of  paradox occurring. 

What none of the models capture is that in actual elections, it is extremely rare to observe a plurality loser paradox without also observing a runner-up paradox. Under spatial models this difference is especially pronounced, as under those models if a plurality loser paradox is exhibited then we almost never also observe a runner-up paradox.  This is generally true across all models, even though the runner-up (RU) paradox is roughly three times as likely to occur as the plurality loser (PL) paradox.

All theoretical models predict that if an election is somewhat ``competitive'' then the probability of demonstrating a reinforcement paradox is large (although usually not larger than 50\%). The empirical results also bear this out, with empirical probabilities sometimes being larger than theoretical ones, even when using theoretical models which often produce close elections. Of course, the empirical probabilities come nowhere close to the probability of 1 predicted by IC, the model which tends to produce extremely competitive elections.

\subsection{Comparison to Previous Reinforcement Paradox Results}

As mentioned in the Introduction, we are aware of only two previous papers \cite{CMMS10,PT14} which estimate the probability with which an election demonstrates a reinforcement paradox. Our work differs from theirs in several significant ways. Most importantly, in most cases the analysis in these earlier papers was based on simulating two separate elections, checking whether the elections had the same winner, and then indicating a reinforcement paradox had occurred if the winner obtained from the combined election was different from the winner in the two generated elections. That is, the authors did not start with one election and look for a partition of the ballots as we do. As Plassmann and Tideman state in \cite{PT14},  ``\emph{We determine the frequencies of the reinforcement paradox by sampling ballots for two separate elections, which we then combine to determine whether the paradox occurred for any voting rule. In contrast, we determine the frequencies of all other paradoxes from a single sample of ballots.}'' Plassmann and Tideman do not try to decompose one set of ballots into a partition as we do, stating  ``\emph{There is no obvious way in which to divide the ballots from a single election into two distinct sets that would yield the appropriate distribution of rankings in the two new sets and would permit us to examine reinforcement.}''  Courtin et al. \cite{CMMS10} generally take a similar approach, simulating two ballot sets and then combining them. We have contributed to the investigation of this paradox by identifying a method in the 3-candidate case, through Propositions \ref{complete_cond} and \ref{partial_cond},  to determine if  such a division is possible, without having to separately test each possible partition of a ballot set into two sub-elections.  We feel our approach is equally valid  since we are interested in the ``paradoxical harm'' felt by a candidate (and perhaps their supporters) in an election exhibiting the paradox: if  there exists a partition of the electorate such that one of the losing candidates can point to each sub-electorate and claim victory, that candidate can reasonably claim some sort of harm when they ultimately lose the overall election. 

A second way our work differs from previous studies is that we use a much larger set of theoretical models as well as  empirical data to conduct our analysis.  
Neither of the previous papers uses any real-world ranked-choice ballot data.\footnote{The paper by Courtin et al. does have a section where they apply their results to parliamentary (or House) elections in France, Germany, and the United States, but the authors do not have access to candidate preference rankings.} Our use of multiple spatial models, multiple Dirichlet models, and a large dataset of real-world elections (combined with nonparametric bootstrapping analysis) greatly expands previous work.

Thirdly, ours is the first work to provide any analysis of partial ballots in the context of reinforcement paradoxes. Previous work was not concerned with partial ballots, presumably because such work did not analyze any real-world IRV elections.

Finally, both previous papers report extremely small probabilities for observing a reinforcement paradox under IRV. Plassmann and Tideman give a probability of approximately 0.02\%  when using an electorate of at least 1,000 voters. Courtin et al. report similarly small probabilities. Using the term ``Young's reinforcement axiom'' (YRA) synonymously with Young's notion of consistency (Definition \ref{definition_consistency}), Courtin et al. state, ``Our results show that ... cases of violations of YRA are rather rare... $[$A$]$lthough all these $[$iterative scoring$]$ rules violate YRA, the frequencies of violation are very rare. Then, since the violation of this axiom is not so frequent, one should not be unduly worried about its theoretical possibility.'' By contrast, we tend to report  large reinforcement paradox probabilities, which are generally much greater than the probabilities of observing other classical paradoxes. The reason for this discrepancy is our difference in approach, as outlined above. We would most likely obtain similarly small probabilities if, for example, we were to look at each real election and search for a reinforcement paradox by choosing a partition of the ballots at random. Our probabilities are higher because we  consider all possible partitions.

\section{Conclusion}\label{conclusion}

Both theoretical and empirical analysis suggests that three-candidate IRV elections are highly susceptible to the reinforcement paradox if we allow for arbitrary partitions. If we  require the two sub-elections to be similarly-sized, for instance,  there will  likely be fewer instances of the paradox; the same is true if  we restrict the sub-elections so that one sub-election is a small perturbation of the original election. For elections where we have additional geographic information, the frequency of paradoxes drops significantly. Future work could examine how the probability of a reinforcement paradox  changes under restrictions of this kind, or to instances involving a larger number of sub-elections. Since reinforcement paradoxes are most serious when the electorate can be partitioned according to some natural variable (geography, race, education level, etc.), focus on this kind of reinforcement paradox is a natural next step. Any real-world investigation of this kind is likely to be limited by a lack of data, however.

Future work could also investigate the structure of the partitions which demonstrate a paradox. When an election is susceptible to a paradox, it can typically be demonstrated using many different partitions. The structure of the space of paradox-producing partitions is therefore worth exploring. For example, how ``similar'' can $\mathcal{B}_1$ and $\mathcal{B}_2$ be made to each other, or how similar can $\mathcal{B}_1$ be made to $\mathcal{B}$? Courtin et al. \cite{CMMS10} and Plassmann and Tideman \cite {PT14} generate $\mathcal{B}_1$ and $\mathcal{B}_2$ using the same probability distribution; do the two ballot sets in real-world partitions seem to come from the same distribution? We could use notions of similarity adapted from \cite{FS19,SB25} to study where $\mathcal{B}_1$ and $\mathcal{B}_2$ lie in a ``map of elections,'' either relative to each other or to $\mathcal{B}$. At the very least, our results highlight how non-convex this landscape can be.


Finally, we focus on IRV in this article because of its use in real-world political elections, but it would be interesting to repeat our analysis for other voting methods which are not positional scoring rules.  For example, many researchers in the social choice space feel that Condorcet consistent methods are superior to IRV (see Chapter 2 of \cite{DFP}, for example), and thus we would welcome an examination of the susceptibility of Condorcet consistent methods to the reinforcement paradox. 

\section*{Acknowledgements} The authors thank the editor and reviewers for comments which greatly improved the article.

\clearpage

\section*{Appendix: District and Precinct Maps}

\begin{figure}[h]
\includegraphics[width=140mm,height=95mm]{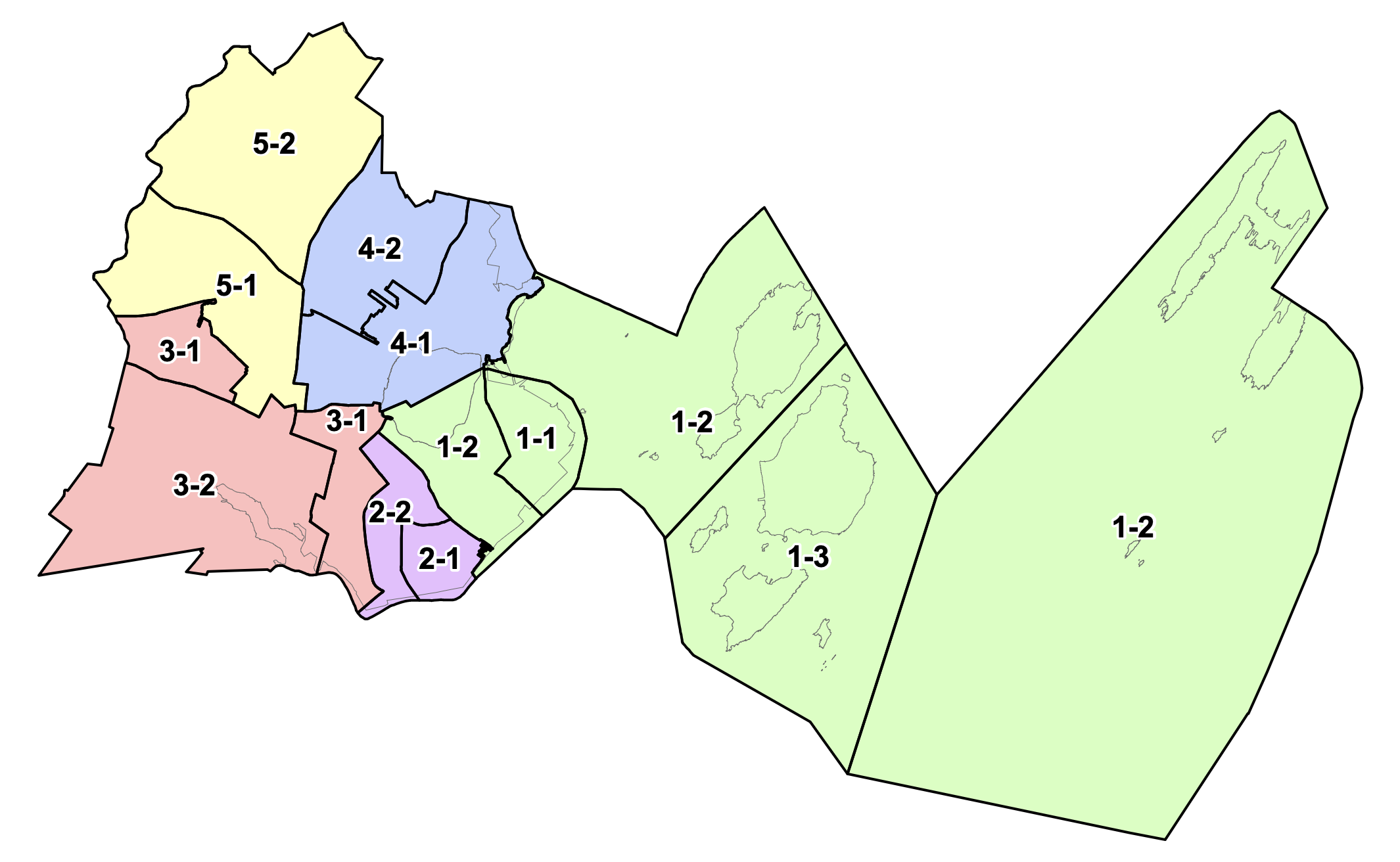}
\caption{A map of the precincts used for the 2021 at-large commissioner election in Portland, ME. Map available at \protect\url{https://www.portlandmaine.gov/292/PDF-Maps}. Accessed 4/1/2026.}
\label{figure:portland_ME}
\end{figure}

\clearpage
\begin{figure}
\centering
\includegraphics[width=115mm]{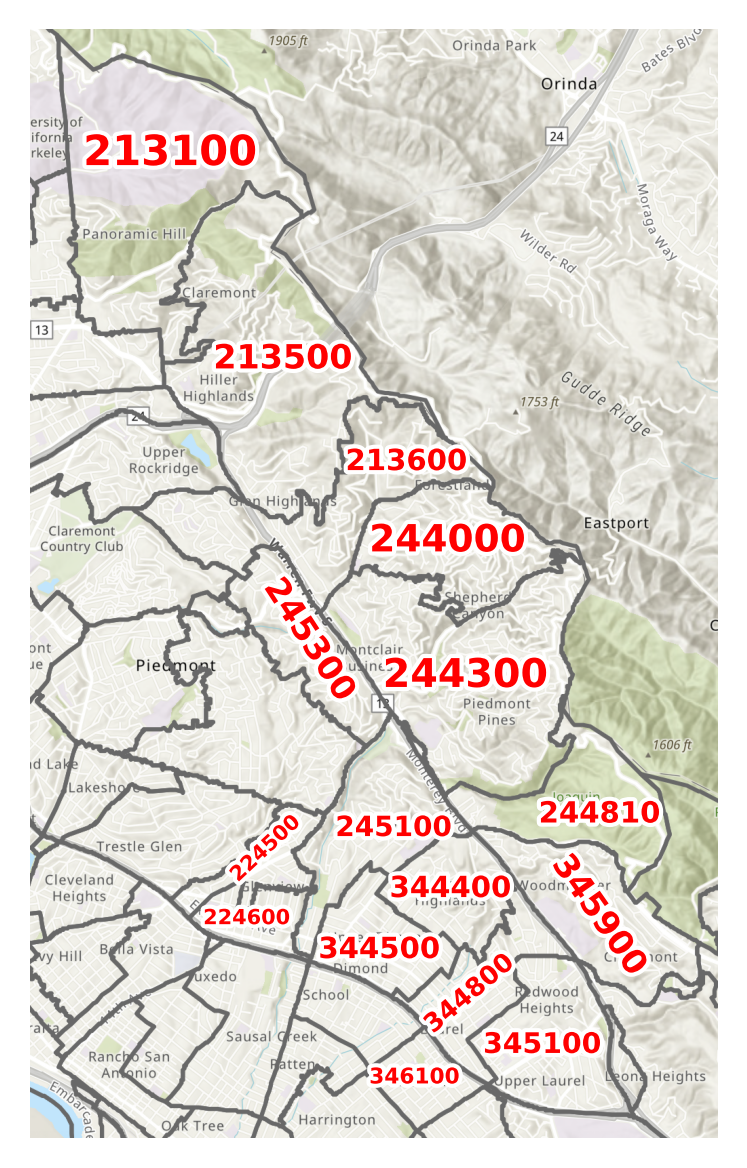}
\caption{A map of the precincts used for the 2022 District 4 Oakland School Director election. Map available at \protect\url{https://experience.arcgis.com/experience/028c197955a94880a2d6b87948d9f2f0/}. Precinct labels added by us. Accessed 4/1/2026.}
\label{figure:Oakland_D4}
\end{figure}

\begin{figure}[h]
\centering
\includegraphics[width=125mm]{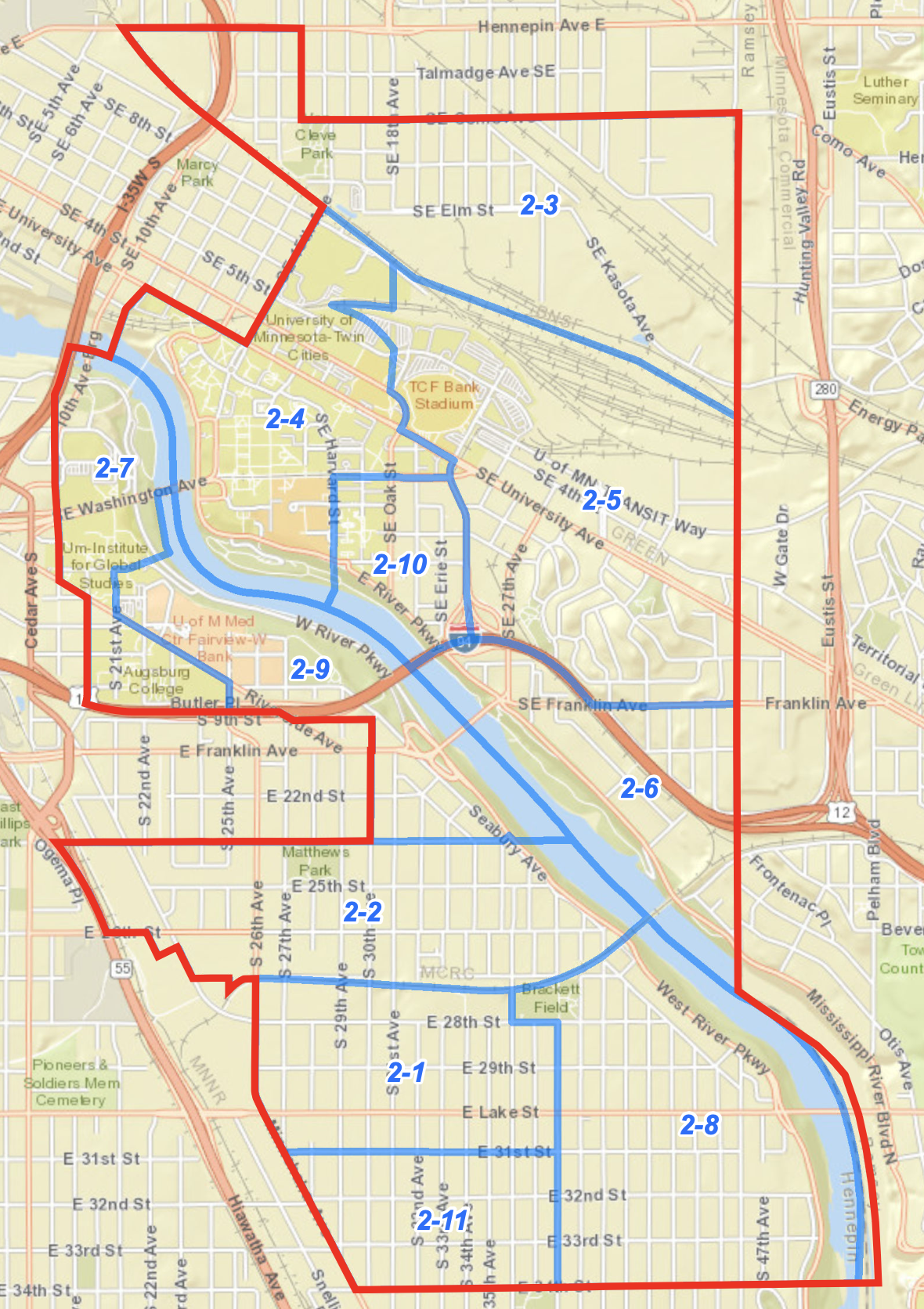}
\caption{A map of the precincts used for the 2021 Ward 2 Minneapolis City Council election. Map provided by the Minneapolis Office of the City Clerk.}
\label{figure:Minneapolis_Ward2}
\end{figure}

\begin{figure}[h]
\centering
\includegraphics[width=120mm,height=75mm]{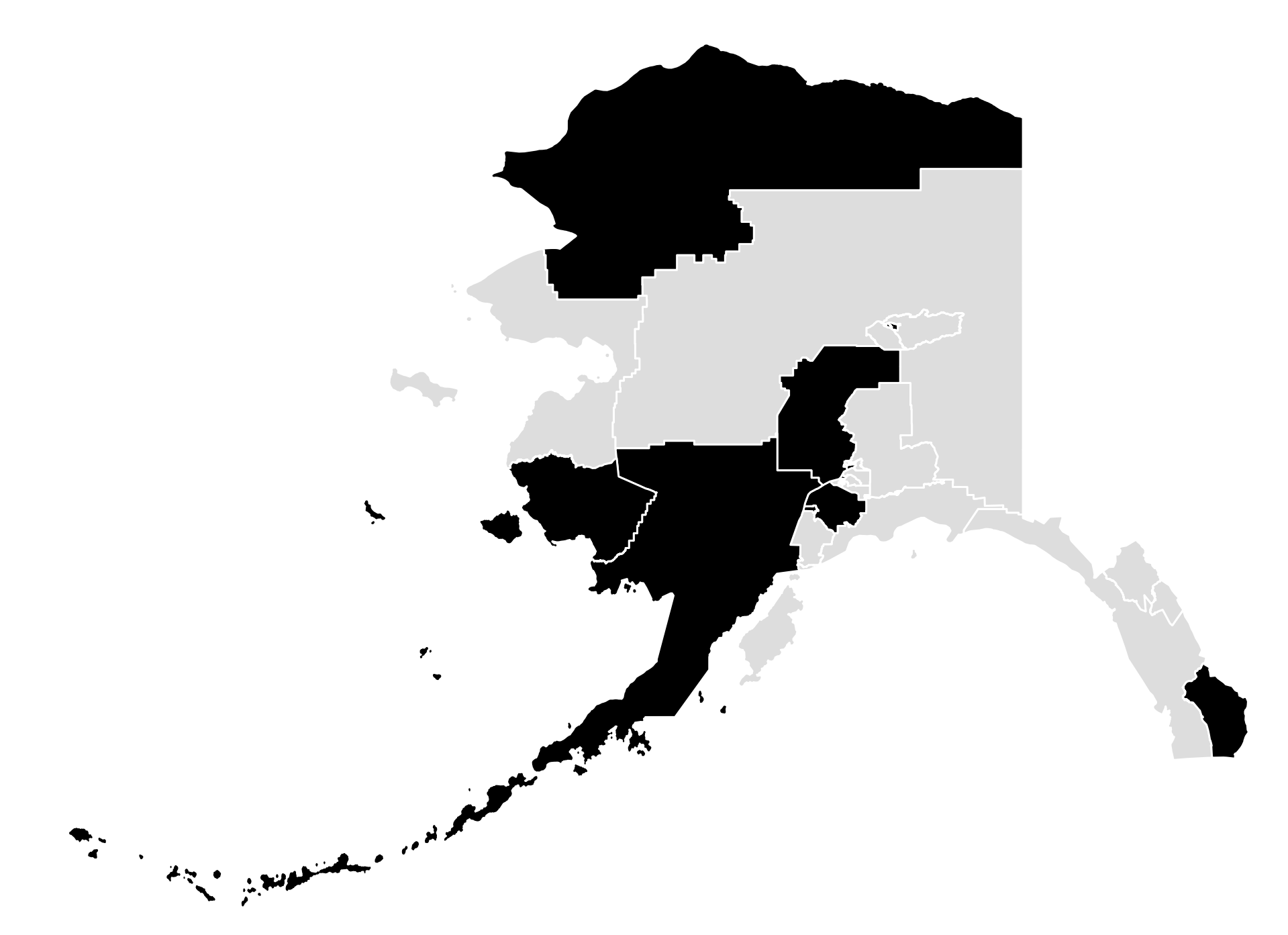}

\includegraphics[width=120mm,height=75mm]{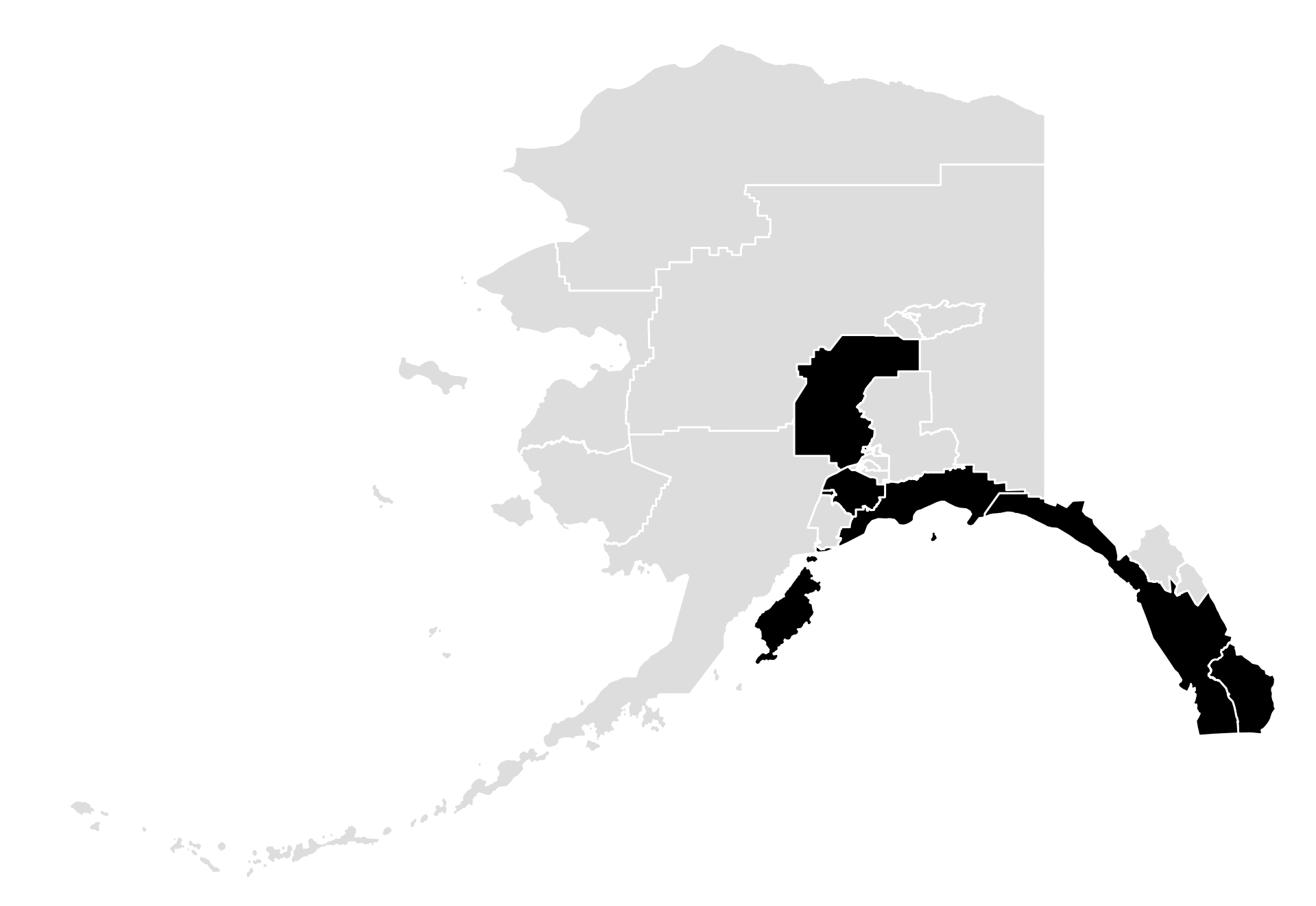}
\caption{(Top) A map showing a partition of state House districts in Alaska, corresponding to the partition in Table \ref{table:AK_partition} demonstrating a plurality loser paradox. Plurality loser Nick Begich wins when using ballots from the black districts plus overseas ballots, and also wins when using ballots from grey districts. The shapefile used to create this map is provided by the Alaska Division of Elections at \protect\url{https://www.akredistrict.org/2022-may-interim-proclamation/}. Accessed 4/9/26. (Bottom) The ``cleanest'' geographic partition we could find. Begich wins the black districts and also wins the gray districts, where we include overseas ballots in the gray.}
\label{figure:AK_map}
\end{figure}

\begin{thebibliography}{}

\bibitem{AK99} F. Aleskerov \& E. Kurbanov. (1999). Degree of manipulability of social choice procedures. In: Alkan (Ed.), \emph{Current Trends in Economics}. Springer.

\bibitem{AM85} R. Aumann \& M, Maschler. (1985) 
Game theoretic analysis of a bankruptcy problem from the Talmud. \emph{Journal of Economic Theory} \textbf{36} (2): 195-213. \url{https://doi.org/10.1016/0022-0531(85)90102-4}

 \bibitem{BY01} M. Balinski \& H. Young. (2001). \emph{Fair Representation: Meeting the Ideal of One Man, One Vote}.  Brookings Institution Press Washington: DC. 
  
\bibitem{BDDW24} G. Benad{\`e}, C. Donnay, M. Duchin, and T. Weighill. (2024) Proportionality for ranked voting, in theory and practice. Preprint: \url{https://mggg.org/PRVTP}.

\bibitem{BMS22} F. Brandt, M. Matth{\"a}us, \& C. Saile. (2022). Minimal voting paradoxes. \emph{Journal of Theoretical Politics} \textbf{34} (4): 527-551.

\bibitem{BBMP24} T. Bardal, M. Brill, D. McCune, \& J. Peters. (2024). Proportional Representation in Practice: Quantifying Proportionality in Ordinal Elections. Submitted.

\bibitem{BGB21} G. Benade, R. Buck, M. Duchin,  D. Gold \& T. Weighill. (2021). Ranked Choice Voting and Proportional Representation. Available at SSRN: https://ssrn.com/abstract=3778021 or \url{http://dx.doi.org/10.2139/ssrn.3778021} 


\bibitem{BW05}
D. Bertsimas and R. Weismantel.
 \emph{Optimization over Integers}.
 Athena Scientific, 2005.
 
\bibitem{BDP24} F. Brandt, C. Dong, and D. Peters. (2024). Condorcet-Consistent Choice Among Three Candidates. \url{https://arxiv.org/abs/2411.19857}.



\bibitem{BISO} W. Bruns, B. Ichim,  C. S{\"o}ger, \& U. von der Ohe. (2022).  \emph{Normaliz. Algorithms for rational cones and affine monoids.}, {\it Normaliz}.   \url{https://www.normaliz.uni-osnabrueck.de}. Accessed Nov. 1, 2024.

\bibitem{BK15} C. Burnett \& V. Kogan. (2015). Ballot (and voter) ``exhaustion'' under Instant Runoff Voting: An examination of four ranked-choice elections. \emph{Electoral Studies} \textbf{37}: 41-49.

\bibitem{C24} L. Cormack. More choices, more problems? Ranked choice voting errors in New York City. \emph{American Politics Research} \textbf{52} (3): 306-319.

\bibitem{CMMS10} S. Courtin, B. Mbih, I. Moyouwou, \& T. Senne. (2010). The reinforcement axiom under sequential positional rules. \emph{Social Choice and Welfare} \textbf{35}: 473-500.

\bibitem{DFP} L. Diamond, E. Foley, \& R. Pildes (Eds). (2025). \emph{Electoral Reform in the United States: Proposals for Combating Polarization and Extremism}. Lynne Rienner Publishers, Boulder, CO.

\bibitem{DMM} M. Dickerson, E. Martin, \& D. McCune. (2024). An Empirical Analysis of the Effect of Ballot Truncation on Ranked-Choice Electoral Outcomes. \emph{The PUMP Journal of Undergraduate Research} \textbf{7}: 79–95. \url{https://doi.org/10.46787/pump.v7i0.3778}.

\bibitem{DDGGHMW} C. Donnay, M. Duchin, J. Gibson, Z. Glaser, A. Hong, M. Mukundan, \& J. Wang. VoteKit: A Python pacakge for computational social choice research. Preprint.

\bibitem{D14} C. Duddy. (2014). Condorcet’s principle and the strong no-show paradoxes. \emph{Theory and Decision} \textbf{77}: 275–285.

\bibitem{EH84} J. Enelow \& M. Hinich. (1984). \emph{The Spatial Theory of Voting: An Introduction.} Cambridge University Press.

\bibitem{EH90} J. Enelow \& M. Hinich. (1990). \emph{Advances in the Spatial Theory of Voting.} Cambridge University Press.

\bibitem{FS19} P. Faliszewski, P. Skowron, A. Slinko, S. Szufa, \& N. Talmon. (2019). How similar are two elections? In
\emph{Proceedings of AAAI-2019}, pages 1909–1916.

\bibitem{FL06} P. Favardin \& D. Lepelley. (2006). Some further results on the manipulability of social choice rules. \emph{Social Choice and Welfare} \textbf{26}: 485-509.

\bibitem{F19} D.S. Felsenthal. (2019). On Paradoxes Afflicting Voting Procedures: Needed Knowledge Regarding Necessary and/or Sufficient Condition(s) for Their Occurrence. In: JF Laslier, H. Moulin, M. Sanver, \& W. Zwicker (eds) \emph{The Future of Economic Design}. Studies in Economic Design. Springer, Cham. \url{https://doi.org/10.1007/978-3-030-18050-8_13}.

\bibitem{FN19} D.S. Felsenthal \& H. Nurmi. (2019). The (In)Vulnerability of 20 Voting Procedures to the Inconsistency Paradox (aka Reinforcement Paradox) in a Restricted Domain. In: \emph{Voting Procedures Under a Restricted Domain}. SpringerBriefs in Economics. Springer, Cham. 

\bibitem{FB83}  P. Fishburn and S. Brams. (1983).  Paradoxes of preferential voting. \emph{Mathematics Magazine}, \textbf{56} (4): 207-214. \url{https://doi.org/10.2307/2689808}.

\bibitem{GSZ20} A. Graham-Squire \& N. Zayatz. (2020). Lack of Monotonicity Anomalies in Empirical Data of Instant-runoff Elections. \emph{Representation} \textbf{57}(4): 565–573. \url{https://doi.org/10.1080/00344893.2020.1785536}.

\bibitem{GS24} A. Graham-Squire. (2024). Conditions for Fairness Anomalies in Instant-Runoff Voting. In M. Jones, D. McCune, J. Wilson (Eds), \emph{Mathematical Analyses of Decisions, Voting and Games}. Contemporary Mathematics, Volume 795. American Mathematical Society.

\bibitem{GSM23} A. Graham-Squire and D. McCune. (2023). An Examination of Ranked-Choice Voting in the United States, 2004-2022. \emph{Representation}, \url{https://doi.org/10.1080/00344893.2023.2221689}.

\bibitem{GA} J. Green-Armytage,  T. Tideman, \& Rafael Cosman. (2016). Statistical evaluation of voting rules. \emph{Social Choice and Welfare} \textbf{46}(1): 183-212.

\bibitem{KMT23} E. Kamwa, V. Merlin, \& F.M. Top. (2023). Scoring Run-off Rules, Single-peaked Preferences and Paradoxes of Variable Electorate.  hal-03143741v2

\bibitem{K93} J. Kelly. (1993). Almost all social choice rules are highly manipulable, but few aren't. \emph{Social Choice and Welfare} \textbf{20}(3): 477-494.

\bibitem{KGF} D.M. Kilgour, J.C. Gr\`{e}goire, \& A.M. Foley. (2020). The prevalence and consequences of ballot truncation in ranked-choice elections. \emph{Public Choice} \textbf{184}: 197-218.


\bibitem{LS09} J. Laatu \& W. Smith. (2009). The rank-order votes in the 2009 Burlington mayoral election. Online; accessed 13-July-2024, \url{https://www.rangevoting.org/JLburl09.txt}.

\bibitem{LCB96} D. Lepelley, F. Chantreuil, \& S. Berg. (1996). The likelihood of monotonicity paradoxes in run-off elections. \emph{Mathematical Social Sciences} \textbf{31}(3): 133-146.

\bibitem{LM01} D. Lepelley \& V. Merlin. (2001). Scoring runoff paradoxes for variable electorates. \emph{Economic Theory} \textbf{14}(1): 53-80.

\bibitem{Preflib} N. Mattei \& T. Walsh. (2013). Preflib: A library for preferences. In \emph{ADT}, 259–270. Springer.

\bibitem{MGS24} D. McCune \& A. Graham-Squire. (2024). Monotonicity anomalies in Scottish local government elections. \emph{Social Choice and Welfare}, \url{https://doi.org/10.1007/s00355-024-01522-5}.

\bibitem{MM23} D. McCune \& L. McCune. (2023). The Curious Case of the 2021 Minneapolis Ward 2 City Council Election. \emph{The College Mathematics Journal}: 1-5. \url{https://doi.org/10.1080/07468342.2023.2212548}

\bibitem{MMLS24} D. McCune, E. Martin, G. Latina, \& K. Simms. (2024). A comparison of sequential ranked-choice voting and single transferable vote. \emph{Journal of Computational Social Sciences} \textbf{7}: 643-670.

\bibitem{MW23} D. McCune \& J. Wilson. Ranked-choice voting and the spoiler effect. (2023). \emph{Public Choice} \textbf{196}: 19–50. \url{https://doi.org/10.1007/s11127-023-01050-3}.

\bibitem{MW24} D. McCune \& J. Wilson. (2024). The Negative Participation Paradox in Three-Candidate Instant Runoff Elections.  \emph{Theory and Decision} \textbf{98}: 537-559.

\bibitem{Mi17} N. Miller. (2017). Closeness matters: monotonicity failure in IRV elections with three candidates. \emph{Public Choice}, \textbf{173}: 91-108.  \url{DOI 10.1007/s11127-017-0465-5.}

\bibitem{M23} D. Mollison. (2023). Fair votes in practice. Preprint: \url{arXiv:2303.15310}

\bibitem{M88} H. Moulin. (1988). Condorcet’s principle implies the no show paradox. \emph{Journal of Economic Theory} \textbf{45}: 53–64.

\bibitem{N99} H. Nurmi. (1999). \emph{Voting Paradoxes and How to Deal with Them}. Springer-Verlag.

\bibitem{ON14} J.T. Ornstein \& R.Z. Norman. (2014). Frequency of monotonicity failure under instant runoff voting: estimates based on a spatial model of elections. \emph{Public Choice} \textbf{161}:1–9.

\bibitem{O22} D. Otis. (2022). Single winner ranked choice voting CVRs. Online; accessed 16-July-2024. \url{https://doi.org/10.7910/DVN/AMK8PJ}, Harvard Dataverse, V13.

\bibitem{PR23} S. Pettigrew \& D. Radley. (2023). Ballot Marking Errors in Ranked-Choice Voting. Preprint: \url{https://ssrn.com/abstract=4670677}.

\bibitem{P01} J. P{\'e}rez. (2001). No Show Paradoxes are a common flaw in Condorcet voting correspondences. \emph{Social Choice and Welfare} \textbf{18}: 601–616

\bibitem{PT14} F. Plassmann \& T.N. Tideman. (2014). How frequently do different voting rules encounter voting paradoxes in three-candidate elections?. \emph{Social Choice and Welfare} \textbf{42}: 31-75.

\bibitem{PPR} S. Popov, A. Popova, \& M. Regenwetter. (2014). Consensus in organizations: Hunting for the social choice
conundrum in APA elections. \emph{Decision}, \textbf{1}(2), 123–146.

\bibitem{P13} W. Poundstone.\emph{Gaming the Vote: Why Elections Aren't Fair (and What We Can Do About It)}. (2008). Hill and Wang, First Edition.

\bibitem{RKKH} M. Regenwetter, A. Kantor, A. Kim, \& M. Ho. (2007). The unexpected empirical consensus among consensus methods. \emph{Psychological Science}, \textbf{18} (7), 629–635.

\bibitem{R23} R. Robinette. (2023). Implications of strategic position choices by candidates. \emph{Constitutional Political Economy} \textbf{34}: 445-457.

\bibitem{S95} D. Saari. (1995). \emph{Basic Geometry of Voting}. Springer.

\bibitem{S23} C.G. Song. (2023). What difference does a voting rule make?. \emph{Constitutional Political Economy} \textbf{34}: 275–285. \url{https://doi.org/10.1007/s10602-022-09375-9}.

\bibitem{S24} N. Stephanopoulos. (2024). Finding Condorcet. Forthcoming in \emph{Washington and Lee Law Review}.

\bibitem{SB25} S. Szufa, N. Boehmer, R. Bredereck, P. Faliszewski, R. Niedermeier, P. Skowron, A. Slink, \& N. Talmon. (2025). Drawing a map of elections. Preprint: 	arXiv:2504.03809.

\bibitem{T11} W. Thomson. (2011). Consistency and its converse: an introduction. \emph{Rev. Econ. Design} \textbf{15}: 257–291. \url{doi 10.1007/s10058-011-0109-z}

\bibitem{TUK} K. Tomlinson, J. Ugander, \& J. Kleinberg. (2022). Ballot Length in Instant Runoff Voting.  Preprint: \url{https://arxiv.org/pdf/2207.08958.pdf}.

\bibitem{W98}
L. Wolsey. \emph{Integer Programming}. Wiley, 1998.

\bibitem{Y74} H.P. Young. (1974). An axiomatization of Borda’s rule. \emph{Journal of Economic Theory} \textbf{9}: 43–52.

\bibitem{Y75} H.P. Young. (1975). Social Choice Scoring Functions. \emph{SIAM Journal on Applied Mathematics}, \textbf{28} (4): 824-838.




\end{thebibliography}
\end{document}